\documentclass{LMCS}

\newcommand{\Omit}[1]{}

\usepackage{amsmath,amssymb,stmaryrd,amsthm,amscd,url}
\usepackage{hyperref}

\usepackage{diagrams}
\newarrow{Into} C--->
\newarrow{Onto} ----{>>}
\newarrow{Dashto} {}{dash}{}{dash}>

\newcommand{\e}{\varepsilon}
\renewcommand{\d}{\delta}
\newcommand{\fullversion}[1]{}

\newcommand{\myomega}{\omega}

\newcommand{\K}{\mathcal{K}}
\newcommand{\Z}{\mathcal{Z}}
\newtheorem{Example}[thm]{Example}
\newtheorem{Examples}[thm]{Examples}
\newenvironment{example}{\begin{Example}}{\end{Example}}
\newenvironment{examples}{\begin{Examples}}{\end{Examples}}
\newcommand{\myparagraph}{\paragraph}

\newcommand{\licsmath}[1]{\[ #1 \]}

\newcommand{\licsmatht}[2]{\begin{enumerate}
\item[] \quad $#1$
\item[] \quad $#2$
\end{enumerate}}
\newcommand{\licsmathtt}[3]{\begin{enumerate}
\item[] \quad $#1$
\item[] \quad $#2$
\item[] \quad $#3$
\end{enumerate}}
\newcommand{\licsmathtttt}[5]{\begin{enumerate}
\item[] \quad $#1$
\item[] \quad $#2$
\item[] \quad $#3$
\item[] \quad $#4$
\item[] \quad $#5$
\end{enumerate}}

\newcommand{\bnf}{\mathrel{::=}}

\newcommand{\N}{\mathbb{N}}
\newcommand{\R}{\mathbb{R}}
\newcommand{\I}{\mathbb{I}}
\newcommand{\Bool}{2}
\newcommand{\Sierp}{\mathcal{S}}
\newcommand{\pN}{\mathcal{N}}

\newcommand{\pBool}{\mathcal{B}}
\newcommand{\If}{\,\mathrel{\operatorname{if}}}
\newcommand{\Then}{\mathrel{\operatorname{then}}}
\newcommand{\Else}{\mathrel{\operatorname{else}}}
\newcommand{\True}{1}
\newcommand{\False}{0}
\newcommand{\domain}[1]{{\D_{#1}}}
\newcommand{\total}[1]{{\T_{#1}}}
\newcommand{\totaleq}[1]{\sim_{#1}}
\newcommand{\quo}[1]{\qq_{#1}}
\newcommand{\qq}{\rho}
\newcommand{\D}{D}

\newcommand{\C}{C}
\newcommand{\T}{T}
\newcommand{\kk}[1]{{\C_{#1}}}

\newcommand{\tproduct}{\sigma \times \tau}
\newcommand{\tfunction}{\sigma \to \tau}
\newcommand{\rh}[1]{\rho_{#1}}
\newcommand{\comp}{\circ}
\newcommand{\siff}{\iff}

\newcommand{\g}{\operatorname{\text{\boldmath{$g$}}}}
\newcommand{\f}{\operatorname{\text{\boldmath{$f$}}}}
\newcommand{\x}{\operatorname{\text{\boldmath{$x$}}}}
\newcommand{\sK}{\operatorname{\text{\boldmath{$K$}}}}
\newcommand{\sF}{\operatorname{\text{\boldmath{$F$}}}}

\newcommand{\sQ}{\operatorname{\text{\boldmath{$Q$}}}}
\newcommand{\sP}{\operatorname{\text{\boldmath{$P$}}}}
\newcommand{\sE}{\operatorname{\text{\boldmath{$E$}}}}

\def\doi{4 (3:3) 2008}
\lmcsheading%
{\doi}
{1--37}
{}
{}
{Jan.~27, 2008}
{Aug.~27, 2008}
{}   

\begin{document}

\author[M.~Escard\'o]{Mart\'\i n Escard\'o}
\address{School of Computer Science, University of Birmingham, B15 2TT, UK}
\email{m.escardo@cs.bham.ac.uk}

\title[Exhaustible sets]{Exhaustible sets in higher-type computation}



\begin{abstract}
  We say that a set is \emph{exhaustible} if it admits algorithmic
  universal quantification for continuous predicates in finite time,
  and \emph{searchable} if there is an algorithm that, given any
  continuous predicate, either selects an element for which the
  predicate holds or else tells there is no example.  The Cantor space
  of infinite sequences of binary digits is known to be searchable.
  Searchable sets are exhaustible, and we show that the converse also
  holds for sets of hereditarily total elements in the hierarchy of
  continuous functionals; moreover, a selection functional can be
  constructed uniformly from a quantification functional.  We prove
  that searchable sets are closed under intersections with decidable
  sets, and under the formation of computable images and of finite and
  countably infinite products.  This is related to the fact,
  established here, that exhaustible sets are topologically compact.
  We obtain a complete description of exhaustible total sets by
  developing a computational version of a topological Arzela--Ascoli
  type characterization of compact subsets of function spaces. We also
  show that, in the non-empty case, they are precisely the computable
  images of the Cantor space. The emphasis of this paper is on the
  theory of exhaustible and searchable sets, but we also briefly
  sketch applications. 
\end{abstract}

\maketitle

\section{Introduction}

\noindent
A wealth of computational problems of interest have the following
form:
\begin{quote}
  \em Given a set $K$ and a property $p$ of elements of $K$, decide
  whether or not \emph{all} elements of $K$ satisfy~$p$.
\end{quote}
For $K$ fixed in advance, this is equivalent to the emptiness problem
for~$p$. One is often interested in suitable restrictions on the
possible syntactical forms of the predicate~$p$ that guarantee that
this problem is decidable (or, less ambitiously, that the
non-emptiness problem is semi-decidable) uniformly in the syntactical
form of~$p$. 
In this work, on the other hand, the emphasis is on the set~$K$ rather
than the predicate~$p$, and we study the case in which $K$ is
infinite. Moreover, $p$ is not assumed to be given syntactically or
via any other kind of intensional information: we only use information
about the input-output relation determined by~$p$ considered as a
boolean-valued function.
In the absence of intensional information, continuity of~$p$ plays a
fundamental role, where $p$ is continuous iff for any $x$ in the
domain of $p$, the boolean value $p(x)$ depends only on a finite
amount of information about~$x$.  We work in the realm of higher-type
computation with continuous functionals, using Ershov--Scott domains
to model partial functionals, and Kleene--Kreisel spaces to model
total functionals~\cite{normann:computer,normann:recursion}.

\pagebreak[3]
We say that the set $K$ is \emph{exhaustible} if the above problem can
be algorithmically solved for any continuous $p$ defined on $K$,
uniformly in~$p$. The uniform dependency on $p$ is formulated by
giving the algorithm the type $(D \to \pBool) \to \pBool$, where $D$
is a domain, $K \subseteq D$, and $\pBool$ is the domain of booleans.
The main question investigated in this work is what kinds of infinite
sets are exhaustible.

Clearly, finite sets of computable elements are exhaustible. What may
be rather unclear is whether there are \emph{infinite} examples.
Intuitively, there can be none: how could one possibly check
infinitely many cases in finite time? This intuition is correct
when~$K$ is a set of natural numbers: it is a theorem that, in this
case, $K$ is exhaustible if and only if it is finite.  This can be
proved by reduction to the halting problem, but there is also a purely
topological argument (Remark~\ref{introduction}).  However, it turns
out that there is a rich supply of infinite exhaustible sets. A first
example, the Cantor space of infinite sequences of binary digits, goes
back to the 1950's, or even earlier, with the work of Brouwer, as
discussed in the related-work paragraph below.

We say that $K$ is \emph{searchable} if there is an algorithm that,
given any continuous predicate~$p$, either selects some $x \in K$ such
that $p(x)$ holds, or else reports that there isn't any. It is easy to
see that searchable sets are exhaustible. We show that, for sets of
total elements, the converse also holds and hence the two notions
coincide.  Moreover, a selection functional can be constructed
uniformly from a quantification functional
(Section~\ref{characterization}).

We develop tools for systematically building exhaustible and
searchable sets, and some characterizations, including the following:
they are closed under intersections with decidable sets, under the
formation of computable images and of finite and countably infinite
products (Section~\ref{building}).  In the case of exhaustibility, the
last claim is restricted to sets of total elements, and is open beyond
this case.  The non-empty exhaustible sets of total elements are
precisely the computable images of the Cantor space (Section~\ref{characterization}). We also formulate and
prove an Arzela--Ascoli type characterization of exhaustible sets of
total elements of function types (Section~\ref{arzela:ascoli}).

The above closure properties and characterizations resemble those of
compactness in topology. This is no accident: we show that exhaustible
sets of total elements are indeed compact, in the Kleene--Kreisel
topology (Section~\ref{criteria}).  This plays a crucial role in the
correctness proofs of some of the algorithms, and, indeed, in their
very construction.  Thus, the specifications of all of our algorithms
can be understood without much background, but an understanding of the
working of some of them requires some amount of topology.  We have
organized the presentation so that the algorithms occurring earlier
are motivated by topology but don't rely on knowledge of topology for
their formulation or correctness proofs.

In Section~\ref{background} we include background material that can be
consulted on demand and in Section~\ref{definitions} we define the
central notions investigated in this work. In Section~\ref{technical}
we include technical remarks, further work, announcement of results,
applications, and research directions.  In the concluding
Section~\ref{conclusion} we review the role topology plays in our
investigation of exhaustible and searchable sets.

\pagebreak[4]
\myparagraph{Related work.}  Brouwer's Fan functional gives the
modulus of uniform continuity of a discrete-valued continuous
functional on the Cantor space.  According to personal communication
by Dag Normann, computability of the Fan functional was known in the
late 1950's. This immediately gives rise to the exhaustibility of the
Cantor space. A number of authors have considered the definability of
the Fan functional in various formal systems.
Normann~\cite{normann:computer} cites Tait (1958, unpublished), Gandy
(around 1982, unpublished) and Berger~\cite{berger:thesis} (1990).
Tait showed that the Fan functional is not definable from Kleene's
schemes S1--S9 interpreted over \emph{total} functionals. Berger
observed that, for \emph{partial functionals}, PCF definability
coincides with S1--S9 definability, and showed that the Fan functional
is PCF definable. In order to do that, he first explicitly defined a
selection functional for the Cantor space.  Then Hyland informed the
community that Gandy was aware of the PCF/S1--S9 definability of the
Fan functional for the partial interpretation of Kleene's schemes,
but Gandy's construction seems to be lost.


\pagebreak[3] \myparagraph{Acknowledgements.} I have benefited from
stimulating discussions with, and questions by, Andrej Bauer, Ulrich
Berger, Dan Ghica, Achim Jung, John Longley, Paulo Oliva, Matthias
Schr\"oder, and Alex Simpson. I also thank Dag Normann for having
answered many questions regarding the history and technical
ramifications of the subject of higher-type computation, and for
sending me a copy of Tait's unpublished manuscript --- but the reader
should consult his paper~\cite{normann:computer} for a more accurate
and detailed account.

\section{Background} \label{background}

The material developed here can consulted on demand, except for
Section~\ref{domains}, which introduces and briefly discusses our
model of computation. Some readers will be more familiar with domain
theory and Ershov--Scott continuous functionals (and PCF or functional
programming) via denotational semantics, and others with the
Kleene--Kreisel continuous functionals via higher-type computability
theory, and we consider these two models and their
relationship~\cite{normann:computer}.  Alternatively, we could have
worked with Weihrauch's model of computation via
representations~\cite{weihrauch:analysis}, which generalizes Kleene's
approach via associates~\cite{normann:recursion}.  Even better, we
could have worked with the QCB model of computation, which subsumes
both domain theory and representation theory in a natural
way~\cite{MR2328287,MR1948051}. We adopt the Ershov--Scott and
Kleene--Kreisel approaches as they have played a wide
role~\cite{normann:computer}. A presentation based on QCB spaces would
have been not only more general but also cleaner in several ways, but
less familiar and perhaps more technically demanding.  Our objective
in this paper is to address the essential issues without getting
distracted by an excessive amount of generality.

\subsection{Domains of computation} \label{domains}

We work on a cartesian closed category of computable maps of
effectively given domains that contains the flat domains of booleans,
\[\pBool=\{\False,\True,\bot\},\]
and of natural numbers, 
\[\pN=\N \cup \{\bot\},\] 
such as~\cite{egli:constable} or~\cite{smyth:effectively} among other
possibilities. Notice that these categories are closed under countable
cartesian powers. We don't need to, and we don't, explicitly refer to
effective presentations, and in particular to numberings of finite
elements or abstract bases etc.\ to formulate computability results.
We instead start from well known computable functions and use the fact
that computable functions are closed under definition by lambda
abstraction, application, least fixed points etc.  Moreover, we don't
invoke non-sequential functions such as Platek's
parallel-or~\cite{scott:lcf} or Plotkin's
parallel-exists~\cite{plotkin:lcf} in order to construct new
computable functions. Our algorithms can thus be directly understood
as functional programs in e.g.\ PCF~\cite{plotkin:lcf},
FPC~\cite{plotkin:domains} (PCF extended with recursive types,
interpreted as solutions of domain equations) or practical versions of
FPC such as Haskell~\cite{bird,haskell:hutton}, as done
in~\cite{escardo:lics07}.

At some point we need further assumptions on our domains of
computation to be able to formulate and prove certain results. Some of
those results can be formulated, and perhaps also be proved, for
domains with totality in the sense of Berger~\cite{berger:total}.  We
consider the particular case consisting of the smallest collection of
domains containing $\pBool$ and $\pN$ and closed under finite
products, countable powers and exponentials (=function spaces).

\subsection{Higher-type computation} \label{higher:background}

The remainder of this section is not needed until
Theorem~\ref{uniform:cont}.  As discussed in
e.g.~\cite{normann:computer,MR2143877,longley:ubiquitous}, there are
many approaches to higher-type computation.  Kleene defined the total
functionals directly, but it has been found more convenient to work
with the larger collection of partial functionals and isolate the
total ones within them, as done by Kreisel. The approaches are
equivalent, and such total functionals are often referred to as
\emph{Kleene--Kreisel functionals} or \emph{continuous functionals}.
It turns out that, as discussed by Normann~\cite{normann:computer},
this coincides with another approach that also arises programming
language semantics: equivalence classes of total functionals on
Ershov--Scott domains.  We work with both total functionals on domains
and a characterization of the Kleene--Kreisel functionals, due to
Hyland, in terms of compactly generated spaces.

\myparagraph{Types.}  The \emph{simple types} are defined by induction
as
\licsmath{\sigma,\tau \bnf o \mid \iota \mid \tproduct \mid \tfunction,}
where $o$ and $\iota$ are ground types for booleans and natural
numbers respectively.  The subset of \emph{pure types} is defined by
\licsmath{\sigma \bnf \iota \mid \sigma \to \iota.}  As usual, we'll
occasionally reduce statements about simple types to statements about
pure types.

\myparagraph{Partial functionals.}  
For each type $\sigma$, define a domain~$\domain{\sigma}$ of
\emph{partial functionals} of type $\sigma$ by induction as follows:
\licsmathtt{\domain{o} = \pBool, \quad \domain{\iota} = \pN,}
{\domain{\tproduct} = \domain{\sigma} \times
  \domain{\tau},}{\domain{\tfunction} = (\domain{\sigma} \to
  \domain{\tau}) = \domain{\tau}^{\domain{\sigma}}} where the products
and exponentials are calculated in the cartesian closed category of
continuous maps of Scott domains, where a \emph{Scott domain} is an
algebraic, bounded complete, and directed complete
poset~\cite{abramsky:jung}.

\myparagraph{Total functionals.} 
For each type $\sigma$, define 
a set $\total{\sigma}\subseteq \domain{\sigma}$ of \emph{total
  functionals} and a relation~$\totaleq{\sigma}$ on~$\domain{\sigma}$
as follows, where $\gamma$ ranges over the ground types~$o$ and~$\iota$:
\[
\begin{array}{lcl}
\total{o} = \Bool = \{\False,\True\}, \quad \total{\iota} = \N,
& \qquad & x \totaleq{\gamma} y \siff x,y \in \total{\gamma}
    \wedge x = y. \\[1ex]
\total{\tproduct} = \total{\sigma} \times \total{\tau},
&  & 
(x,x') \totaleq{\tproduct} (y,y') \siff \text{$x \totaleq{\sigma} y \land x' \totaleq{\tau} y'$}, \\[1ex]
\total{\tfunction} = \{ f \in \domain{\tfunction} \mid f(\total{\sigma}) \subseteq
\total{\tau} \},
& & 
{f \totaleq{\tfunction} g \siff \forall  x \totaleq{\sigma} y. f(x) \totaleq{\tau} g(y).}
\end{array}
\]
Then the set $\total{\sigma}$ can be recovered from the
relation~$\totaleq{\sigma}$ as 
\[ x \in \total{\sigma} \siff x \totaleq{\sigma} x, \]
and the relation can be recovered from the set as
\begin{eqnarray*}
x \totaleq{\sigma} y & \siff & x \sqcap y \in \total{\sigma} \\
& \siff & x,y
\in \total{\sigma} \text{ and $x$ and $y$ are bounded above}.
\end{eqnarray*}
See e.g.~\cite{berger:total} and~\cite{plotkin:totality}. In
particular, $\totaleq{\sigma}$ is an equivalence relation
on~$\total{\sigma}$.

\myparagraph{Computability.}  Plotkin~\cite{plotkin:lcf} characterized
the computable partial functionals as those that are PCF-definable
from parallel-or and parallel-exists~\cite{plotkin:lcf}.  All
computable functionals we construct from Section~\ref{characterization} onwards
are defined in PCF without parallel extensions. This characterization
of computability includes, in particular, total functionals. An
interesting fact, which we don't need to invoke, is that every
\emph{total} functional definable in PCF with parallel extensions is
equivalent to one definable in PCF without parallel
extensions~\cite{normann:totality}.

\subsection{Kleene--Kreisel functionals} \label{kk:background}

For each type~$\sigma$, define by induction a set~$\kk{\sigma}$ of
\emph{Kleene--Kreisel functionals} of type~$\sigma$ and a surjection
$\rh{\sigma} \colon \total{\sigma} \to \kk{\sigma}$ as follows, so
that
\licsmath{\kk{\sigma} \cong \total{\sigma}/\totaleq{\sigma}.}
For ground types and product types, define
\[
\begin{array}{ll}
\kk{o} = 2 = \total{o}, \quad \kk{\iota} = \N = \total{\iota}, \quad & \rh{\gamma}(x)=x. \\
\kk{\tproduct} = \kk{\sigma} \times \kk{\tau}, \quad & \rh{\tproduct} = \rh{\sigma} \times \rh{\tau}.
\end{array}
\]
For function types, consider the diagram \\[0.1ex]
\licsmath{ (\dagger) \quad
\begin{diagram}
   \domain{\sigma} & \lInto & \total{\sigma} & \rOnto_{\rh{\sigma}} & \kk{\sigma} \\
   \dTo^{f} & (1) & \dDashto & (2) & \dDashto_{\phi} \\
   \domain{\tau} & \lInto & \total{\tau} & \rOnto^{\rh{\tau}} & \kk{\tau}. \\
\end{diagram}
} 

\medskip
\noindent
The square (1) commutes for some map $\total{\sigma} \to \total{\tau}$
if and only if $f \in \total{\tfunction}$, and in this case the map is
uniquely determined as the (co)restriction of~$f$.  Moreover, in this
case, there is a unique map $\phi$ making the square (2) commute,
because $\rh{\sigma}$ is a surjection. 
We define
\begin{eqnarray*}
\kk{\tfunction} & = & \{ \phi \colon \kk{\sigma} \to \kk{\tau}
  \mid \exists f \in \total{\tfunction}.\text{$(2)$ commutes}\}, \\
\rh{\tfunction}(f) & = & \text{the unique $\phi$ such that $(\dagger)$
    commutes.}
\end{eqnarray*}
Then, by construction, for any $\sigma$ and all $x,y \in
\domain{\sigma}$, we have that 
\[\text{$x \totaleq{\sigma} y$ iff $x,y
  \in \total{\sigma}$ and $\rh{\sigma}(x)=\rh{\sigma}(y)$}.\]
If $(\dagger)$ commutes, we say that $f$ is a \emph{representative} of
$\phi$. A Kleene--Kreisel functional is computable iff it has a
computable representative.

\begin{lem} \label{pure} 
Every $\kk{\sigma}$ is a computable retract of
$\kk{\tau \to \iota}$ for some $\tau$.  
\end{lem}
A stronger form of this is known as ``simple types are retracts of
pure types'' (see e.g.~\cite{longley:ubiquitous}). Here we use the
fact that every pure type is either $\iota$ or of the form $\tau \to
\iota$, and that $\iota$ is a retract of $\tau \to \iota$ for any
$\tau$.

\subsection{Compactly generated spaces} \label{compactly}

The remainder of this section is not needed until
Section~\ref{criteria}, where it is used in order to formulate and
prove the crucial Lemma~\ref{lemma:criterion} that establishes
compactness of exhaustible sets of total elements.  Compactly
generated Hausdorff spaces, or $k$-spaces, to be introduced shortly,
are to total computation as Scott domains are to partial computation.
This becomes clear in Section~\ref{hyland}. Here we briefly introduce
$k$-spaces and some of their fundamental topological properties that
are applied to prove various computational theorems. For more details
and proofs, see e.g.~\cite{escardo:lawson:simpson} or the references
contained therein.

We begin by considering the Hausdorff case. If $F$ is a closed set of a
Hausdorff space~$X$, then $K \cap F$ is closed for every compact set
$K \subseteq X$.  Any set $F$ that satisfies this condition is called
$k$-closed. The Hausdorff space $X$ is a \emph{$k$-space} iff every
$k$-closed set is closed. (This is equivalent to saying that $X$ is
the colimit of its compact subspaces ordered by inclusion.) Any
Hausdorff space can be transformed into a $k$-space by stipulating
that all $k$-closed sets are closed.  In categorical terms, this
construction is a coreflection of the category of Hausdorff spaces
into its subcategory of $k$-spaces. 

The category of Hausdorff $k$-spaces is cartesian closed. Given
objects $X$ and $Y$, their categorical product $X \times Y$ is the
coreflection of their topological product. Their exponential $Y^X$
consists of the continuous maps $X \to Y$ under the coreflection of
the compact-open topology.  The compact-open topology has subbasic
open sets of the form \[N(K,V) = \{ f \in Y^X \mid f(K) \subseteq
V\},\] where $K$ is a compact subset of $X$ and $V$ is an open subset
of $V$.  Cartesian closedness amounts to the fact that the evaluation
map
\begin{eqnarray*}
  Y^X \times X  & \to & Y \\
  (f,x) & \mapsto & f(x)
\end{eqnarray*}
is continuous, and that for any continuous map $f \colon Z \times X \to Y$,
its transpose
\begin{eqnarray*}
  \bar{f} \colon Z & \to&  Y^X \\
  z & \mapsto & (x \mapsto f(z,x))
\end{eqnarray*}
is continuous. Equivalently, $f$ is continuous iff $\bar{f}$ is
continuous.

We also need to consider $k$-spaces without the restriction to the
Hausdorff case. Let $X$ be an arbitrary topological space. A
\emph{probe} is a continuous function $p \colon K \to X$ where $K$ is
a compact Hausdorff space. A set $F \subseteq X$ is $k$-closed if
$p^{-1}(F)$ is closed for every probe $p \colon K \to X$.  Then,
again, $X$ is a $k$-space iff every $k$-closed set is closed, and
$k$-spaces form a cartesian closed category, and the inclusion of
Hausdorff $k$-spaces preserves products and exponentials.  All locally
compact spaces are $k$-spaces, and this includes non-Hausdorff
examples such as Scott domains under the Scott topology.  The
description of the products and exponentials in the general case is
omitted and the reader is referred to the above references, but in any
case they are not needed for the purposes of this work, with one
exception: an exponential of $k$-spaces whose base is the Sierpinski
space has the Scott topology~\cite{escardo:lawson:simpson}. This is
applied in the proof of the following lemma.

Denote by $\Sierp$ the Sierpinski space with an isolated point $\top$
and a limit point $\bot$. This is the same as the domain
$\{\bot,\top\}$ under the Scott topology. For any topological space
$X$, a function $p \colon X \to \Sierp$ is continuous iff
$p^{-1}(\top)$ is open, and a set $U \subseteq X$ is open iff its
characteristic function $\chi_U$, defined by $\chi_U(x)=\top \iff x
\in U$, is continuous. Thus, using the Sierpinski space, the notion of
openness is reduced to that of continuity.
The following reduces the notion of compactness to that of
continuity (a particular case of this is proved
in~\cite{escardo:barbados}, with essentially the same proof as the one
give here).

\pagebreak[3]
\begin{lem} \label{crucial:generalized}
  If $X$ is a $k$-space, a set $K \subseteq X$ is
  compact if and only if the universal quantification functional
  $\forall_K \colon \Sierp^X \to \Sierp$ defined by
\begin{quote}
$\forall_K(p)=\top$ iff $p(x)=\top$ for all $x \in K$
\end{quote}
is continuous. 
\end{lem}
\begin{proof}
  A set $K$ is compact if and only if every directed cover of $K$ by
  open sets has a member that covers $K$, because from any cover one
  obtains a directed cover with the same union by adding the finite
  unions of the members of the cover.  Hence by definition of the
  Scott topology, a set is compact if and only if its
  open-neighbourhood filter is open in the Scott topology of the
  lattice of open sets.  But $U \mapsto \chi_U$ is a bijection from the
  lattice of open sets to the points of $\Sierp^X$, and it was shown
  in~\cite{escardo:lawson:simpson} that the topology of the
  exponential $\Sierp^X$ is the one induced by this bijection. Hence
  the functional $\forall_K$ is continuous iff $\forall_K^{-1}(\top)$
  is open iff the set of characteristic functions $\chi_U$ with $K
  \subseteq U$ is open iff the open neighbourhood filter of $K$ is
  open iff $K$ is compact.
\end{proof}

\subsection{Hyland's characterization of the Kleene--Kreisel
  functionals} \label{hyland} For certain constructions and proofs of
algorithms, we consider a topology on the set of Kleene--Kreisel
functionals.
\begin{defi}
  Endow $\total{\sigma}$ with the relative Scott topology and
  $\kk{\sigma}$ with the quotient topology of the
  surjection~$\rh{\sigma} \colon \total{\sigma} \to \kk{\sigma}$.  We
  refer to this topology on $\kk{\sigma}$ as the \emph{Kleene--Kreisel
    topology}, and to the resulting spaces $\kk{\sigma}$ as the
  \emph{Kleene--Kreisel spaces}. The points of the Kleene--Kreisel
  spaces are often referred to as \emph{the continuous functionals} in
  the higher-type computability (or higher-type recursion) literature.
  \qed
\end{defi}
A proof of the following inductive topological characterization of the
Kleene--Kreisel spaces, attributed to Hyland, can be found in
Normann~\cite{normann:recursion}.
\begin{lem} 
  With products and exponentials in the category of
  Hausdorff $k$-spaces,
\begin{enumerate}
\item $\kk{\gamma}$ has the discrete
topology for~$\gamma$ ground, 
\item  $\kk{\tproduct} = \kk{\sigma} \times \kk{\tau}$ and 
\item $\kk{\tfunction}
  = {\kk{\tau}}^{\kk{\sigma}}$.
\end{enumerate}

\end{lem}

The following two lemmas, which are part of the folklore of the
subject, are applied in order to show that exhaustible sets of total
elements are compact in the Kleene--Kreisel topology
(Lemma~\ref{lemma:criterion}(\ref{lemma:criterion:1})).
A set is called \emph{clopen} if it is both closed and open. 
\begin{lem} \label{clopen:extension} For every clopen $U \subseteq
  \kk{\sigma}$ there is a total predicate $p \in (\domain{\sigma}
  \to \pBool)$ such that $\quo{\sigma}^{-1}(U) \subseteq
  p^{-1}(\True)$ and $\quo{\sigma}^{-1}(\kk{\sigma} \setminus U) \subseteq
  p^{-1}(\False)$.
\end{lem}
\begin{proof}
  Because $U$ is clopen, its characteristic function $\chi_U \colon
  \C_\sigma \to \Bool$ is continuous, and hence so is the composite $i
  \comp \chi_U \comp \quo\sigma \colon \total{\sigma} \to \pBool$,
  where $i \colon \Bool \to \pBool$ in the inclusion. Because $\T$ is
  dense in $\D_\sigma$ (see e.g.\ \cite{berger:total}) and because
  Scott domains, and hence $\pBool$, are densely injective (see e.g.\
  \cite{gierz:domains}), by definition of injectivity this extends to
  a continuous function $p \colon \domain{\sigma} \to \pBool$. Then
  $p$ is total by construction, and the extension property amounts to
  the above set inclusions.
\end{proof}
A space is \emph{zero-dimensional} iff it has a base of clopen sets.
The \emph{zero-dimensional reflection} $\Z \C$ of a space~$\C$ is
obtained by taking the same set of points and the clopen sets as a
base.

\begin{lem} \label{zero:compact} $\Z \kk{\sigma}$ and $\kk{\sigma}$
  have the same compact subsets. 
%
\end{lem}
\begin{proof}
  We first show that $\K \Z C=C$ where $C=\kk{\sigma}$ and $\K$ is the
  coreflector into the category of $k$-spaces.  The property $\K \Z C
  = C$ is easily seen to be inherited by retracts, and hence, by
  Lemma~\ref{pure}, it is enough to consider $\sigma=\tau \to \iota$,
  and hence $C=\N^Y$ for some $k$-space~$Y$.  Exponentials in
  $k$-spaces are given by the $k$-coreflection of the compact-open
  topology on the set of continuous maps. When the target is $\N$, the
  compact-open topology is clearly zero-dimensional and Hausdorff.
  Now, it is easy to see that $\K \Z \C = \C$ iff there is some
  zero-dimensional topology whose $k$-reflection is $\C$, and hence we
  are done. The result then follows from the well-known fact that a
  Hausdorff space has the same compact sets as its $k$-coreflection.
\end{proof}
If the spaces $\kk{\sigma}$ were zero-dimensional, the above lemma
would be superfluous. But Matthias Schr\"oder~\cite{schroeder:not:regular} has recently shown,
after this paper was produced and refereed, that the spaces
$\kk{\sigma}$ are not zero-dimensional, and in fact not even regular,
answering a question of~\cite{bauer:escardo:simpson,normann:zero}.

\section{Exhaustible and searchable sets} \label{definitions}

We now formulate the central notions investigated in this work.
\begin{defi} \label{definedon} if $K$ is a subset of the domain
  $\D$, we say that a predicate $p \in (\D \to \pBool)$ is
  \emph{defined on~$K$} if $p(x) \ne \bot$
  for every $x \in K$. \qed
\end{defi}

\begin{defi} \label{exhaustible:def} We say that a subset $K$ of
  the domain $D$ is \emph{exhaustible} if there is a computable
  functional $\forall_K \colon (\D \to \pBool) \to \pBool$ such that
  for any $p \in (\D \to \pBool)$ defined on $K$,
  \licsmath{
  \forall_K(p)=
  \begin{cases}
    \True & \text{if $p(x)=\True$ for all $x \in K$,} \\
    \False & \text{if $p(x)=\False$ for some $x \in K$.} \\
  \end{cases}
  }
  Such a \emph{universal quantification functional} is not uniquely
  determined, because its behaviour is not specified for predicates
  $p$ that are not defined on $K$. For the sake of clarity, we'll
  often write ``$\forall_K(\lambda x.  \dots)$'' as ``$\forall x \in
  K. \dots$''.  \qed
\end{defi}
Clearly, it is equivalent to instead require the existence of a
computable functional $\exists_K \colon (\D \to \pBool) \to \pBool$
such that for any $p \in (\D \to \pBool)$ defined on $K$,
\licsmath{
\exists_K(p)=
\begin{cases}
  \True & \text{if $p(x)=\True$ for some $x \in K$,} \\
  \False & \text{if $p(x)=\False$ for all $x \in K$,} \\
\end{cases}
} 
because such functionals are inter-definable by the De Morgan Laws
and hence we'll freely switch between them.

We now formulate searchability in a way slightly different from that
of the introduction, which is more convenient for our purposes. The
only essential difference is that the present, official definition
excludes the empty set (cf.\ Remark~\ref{searchable:equivalent}).

\pagebreak[3]
\begin{defi} \label{searchable:def}
  We say that a set $K \subseteq \D$ is \emph{searchable} if there is
  a computable functional $\e_K \colon (\D \to \pBool) \to
  \D$ such that, for every predicate $p \in (\D \to \pBool)$ defined
  on~$K$,
\pagebreak[3]
  \begin{enumerate}
  \item $\e_K(p) \in K$, and
  \item \label{searchable:def:2} $p(\e_K(p))=\True$ if $p(x)=\True$ for some $x \in K$.
  \end{enumerate}
  Again, notice that the \emph{selection functional} $\e_K$ is not
  uniquely determined by $K$. \qed
\end{defi}
Thus, $\e_K(p)$ is an example of an element of $K$ for
which~$p$ holds, if such an element exists, or a counter-example
in~$K$ if no such example exists.
\begin{lem} \label{implicit}
  Searchable sets are exhaustible. 
\end{lem}
\begin{proof}
  Define $\exists_K(p)=p(\e_K(p))$.
\end{proof}
The empty set is exhaustible with $\forall_{\emptyset}(p)=\True$, but
it is not searchable because the condition $\e_{\emptyset}(p)
\in \emptyset$ cannot hold. But we'll see in
Section~\ref{characterization} that, under fairly general and natural
conditions, the two notions turn out to agree in the non-empty case.
Moreover, it is clear that non-empty finite sets of computable
elements are both exhaustible and searchable.

\begin{rem} \label{searchable:equivalent}
With $1=\{\star\}$, an equivalent definition of searchability is that
\begin{enumerate}
\item $K$ has a computable element $e_K$, and
\item there is $\e'_K \colon (\D \to
  \pBool) \to 1+\D$ computable such that $\e'_K(p)
  = \star$ if there is no example, and otherwise $\e'_K(p)
  \in K$ and $p(\e'_K(p))=\True$.
\end{enumerate}
In fact, given $\e_K$ one can define $e_K =
\e_K(\lambda x.\True)$ and \licsmath{\e'_K(p) = \If
  p(\e_K(p)) \Then \e_K(p) \Else \star.}
Conversely, given $\e'_K$ and $e_K$ as specified, one can
define \licsmath{\e_K(p)= \linebreak[3] \If
  \e'_K(p)=\star \Then e_K \Else \e'_K(p).} \qed
\end{rem}

Regarding examples, we'll deduce later the known fact that the
Cantor space is searchable.  For the moment, we show that the natural
numbers with a point at infinity form a searchable set.
\begin{defi}
  The \emph{one-point compactification of the natural numbers} is the
  subspace $\N_\infty$ of the Cantor space $2^\omega \subseteq
  \pBool^\omega$ consisting of the sequences $0^n 1^\omega$
  (representing natural numbers~$n$) and $0^\omega$ (representing the
  added point at infinity).  \qed
\end{defi}
The relative Scott topology on the Cantor space agrees with the
product topology of the discrete space~$\Bool$, but such topological
considerations are not needed until Section~\ref{criteria}. In
constructive mathematics, $\N_\infty$ is equivalently defined as the
set of sequences $\alpha \in 2^\omega$ with $\alpha_i \le
\alpha_{i+1}$, to avoid excluded middle. In functional programming,
$\N_{\infty}$ also arises as the set of maximal elements of the domain
of lazy natural numbers.
\begin{example}
  $\N_\infty$ is searchable, with selection functional
  $\e_{\N_\infty}$ defined by primitive recursion as
\[
\e_{\N_\infty}(p)(i) = \exists n \le i.\, p(0^n 1^\omega).
\]
Notice that
$\e_{\N_\infty}(p)$ is the infimum of the set of solutions
$\alpha \in \N_\infty$ of $p(\alpha)=\True$, including the case in
which the set is empty, for which $\e_{\N_\infty}(p)=\infty$.
\qed
\end{example}

This construction is implicit in Exercise 1 of Barendregt~\cite[Page
581]{barendregt}, attributed to Kreisel.  The point of that exercise
is that this algorithm can be interpreted as a functional in the full
type hierarchy, defined in G\"odel's system~$T$, that also works for
discontinuous~$p$. The exercise uses this to prove that the
substructure of definable elements is not extensional, or
equivalently, that the set-theoretical model of system~$T$ fails to be
fully abstract. This exercise was brought to my attention by Gordon
Plotkin and Alex Simpson, after I posed this full abstraction question
to them. Notice that the Kleene--Kreisel model of system~$T$ is fully
abstract, using the fact that the elements of a dense set are
definable.

\section{Building new searchable sets from old} \label{building}

In this section we develop algorithms that don't require knowledge of
topology but are motivated by topological considerations.  Starting
from the finite sets, the algorithms allow us to systematically build
plenty of infinite searchable sets.
The intuition behind the topological notion of compactness is that
compact sets behave, in many relevant respects, as if they were
finite.  Infinite sets that admit exhaustive search in finite time
share the same intuition. Hence it is natural to conjecture that they
also share similar structural properties. For example, compact sets
are closed under the formation of products (Tychonoff theorem).
Motivated by this, in this section we show that searchable sets are
closed under countable products, and we also export other closure
properties from topology to computation.

\begin{defi} \label{decidableonk} For a given set $K \subseteq
  D$, we say that a set $F \subseteq K$ is \emph{decidable on $K$} if
  there is a computable map $\psi_F \colon D \to \pBool$ defined on
  $K$ such that, for all $x \in K$, $\psi_F(x)=\True$ iff $x \in F$. \qed
\end{defi}
\pagebreak[3]
\begin{prop} \label{prop:intersec} Let $K \subseteq D$ and let $F \subseteq
  K$ be decidable on~$K$.
\begin{enumerate}
\item If $K$ is exhaustible then so is $K \cap F$. 
\item If $K$ is searchable then so is $K \cap F$, provided it is
  non-empty.
\end{enumerate}
\end{prop}
\begin{proof}
  Define \\
$\exists_{K \cap F}(p) = \exists x \in K. x \in F
  \wedge
  p(x)$, \\
$\e_{K \cap F}(p) = \If \exists x \in K.\,
  x \in F \wedge p(x) \Then \e_K(\lambda x. x \in F \wedge
  p(x)) \Else \e_K(\lambda x. x \in F)$.
\end{proof}
The topological motivation for the above proposition is that the
intersection of a closed set with a compact set is compact. Decidable
sets correspond to sets that are open and closed, and hence, bearing
in mind that exhaustible sets (ought to) correspond to compact sets,
the above proposition ought to be true, which it is.  It is an easy
exercise to show that exhaustible and searchable sets are closed under
binary unions.  But binary intersections are problematic.  In fact,
in topology, in the absence of assumptions such as the Hausdorff
separation axiom, compact sets fail to be closed under binary
intersections. Hence any algorithm for binary intersections would have
to exploit specialized topological and/or order-theoretic properties
of domains.  The topological motivation for the following proposition
is that, in topology, continuous images of compact sets are compact.
In fact, it arises by replacing continuity by computability and
compactness by exhaustibility.

\begin{prop} \label{image} Exhaustible and searchable sets are closed
  under the formation of computable images.
\end{prop}
\begin{proof}
  Let $f \in (\D \to \D')$ be computable and let $K$ be a subset of
  $D$.  For any quantification functional $\forall_K \colon (D \to
  \pBool) \to \pBool$, the functional $\forall_{f(K)} \colon (D' \to
  \pBool) \to \pBool$ defined by
\licsmath{\forall_{f(K)}(q)=\forall x \in K.q(f(x))}
is clearly a quantification functional for $f(K)$.  

For any selection functional $\e_K \colon (D \to \pBool) \to
D$, the following definition gives a selection functional
$\e_{f(K)} \colon (D \to \pBool) \to D$:
\licsmath{ \e_{f(K)}(q) =
  f(\e_K(\lambda x.q(f(x))).  }
That is, first find $x$ such that $q(f(x))$ holds, using
$\e_K$, and then apply $f$ to this~$x$.
\end{proof}

The following corresponds to the fact that compact sets in topology
are closed under finite products:
\begin{prop} \label{fin:prod}
  Exhaustible and searchable sets are closed under the formation of
  finite products.
\end{prop}
\begin{proof}
  For $K \subseteq \D$ and $K' \subseteq \D'$ exhaustible, define
  \licsmath{\forall_{K \times K'}(p)=\forall x \in K.\forall x' \in
    K'.p(x,x').}
  For $K \subseteq \D$ and $K' \subseteq \D'$ searchable, to compute
  $\e_{K \times K'}(p)$ we first find $x \in K$ such that
  there is $x' \in K'$ with $p(x,x')$, and then find $x' \in K'$ such
  that $p(x,x')$, i.e.\
\licsmatht{x = \e_K(\lambda x.\exists x' \in K'.p(x,x')),}
{x' = \e_{K'}(\lambda x'. p(x,x')),}
using the fact that searchable sets are exhaustible,
and let $\e_{K \times K'}(p) = (x,x')$.
\end{proof}

Compact sets in topology are closed under arbitrary products.  We now
show that searchable sets are closed under countable products.  We
would like to show that for any sequence of searchable sets $K_i
\subseteq D_i$, their product $\prod_i {K_i} \subseteq \prod_i D_i$ is
also searchable, but this would require dependent types, which are not
part of the traditional higher-type computation formalism (but
see~\cite{berger:dependent} and~\cite{berger:dependent:1}).  So we
assume that the components~$K_i$ of the product are all subsets of the
same domain~$\D$, so that $\prod_i K_i \subseteq \D^\myomega$ instead,
leaving the more general question for future work.

\pagebreak[3]
Given selection functionals
\licsmath{\e_{K_i} \in ((\D \to \pBool) \to \D),}
we wish to construct a selection functional 
\licsmath{\text{$\e_{\prod_i K_i} \in ((\D^\myomega \to \pBool)
    \to \D^\myomega)$}.}  
The idea, which iterates the proof of Proposition~\ref{fin:prod}, is to
let
\licsmath{\e_{\prod_i K_i}(p)=
  x_0x_1x_2\dots x_n\dots,}
where
\licsmathtttt{\text{$x_0 \in K_0$ is such that $\exists \alpha \in \prod_{i} K_{i+1}.p(x_0\alpha)$,}}
{\text{$x_1 \in K_1$ is such that  $\exists \alpha \in \prod_{i} K_{i+2}.p(x_0x_1\alpha)$,}}
{\dots}
{\text{$x_n \in K_n$ is such that $\exists \alpha\in \prod_{i} K_{i+n+1}.p(x_0x_1\dots x_n\alpha)$,}}
{\dots}
The component $x_n$ will be found using $\e_{K_n}$, and
existential quantifications will be recursively reduced to search.  To
make this precise, we change notation.  Given a sequence
\licsmath{\e \in ((\D \to \pBool) \to \D)^\myomega,}
such that $\e_i$ is a selection functional for $K_i$, we wish
to find
\licsmath{\Pi(\e) \in (\D^\myomega \to \pBool) \to \D^\myomega}
that is a selection functional for~$\prod_i K_i$. That is, we are
looking for a computable functional
\licsmath{\Pi \colon ((\D \to
  \pBool) \to \D)^\myomega \to ((\D^\myomega \to \pBool) \to \D^\myomega)}
that transforms any sequence of selection functionals for subsets
of~$\D$ into a selection functional for a subset of~$\D^\myomega$:
%
\licsmath{\Pi(\e)(p)(n) = \text{$x_n$ such that\ $\exists \alpha\in \prod_{i} K_{i+n+1}.p(x_0x_1\dots x_n\alpha)$.}}
%
To complete the derivation of the functional $\Pi$, we reduce the
existential quantification to a suitable recursive call to~$\Pi$.  If
the functional~$\Pi$ is to meet its specification, $\Pi(\lambda
i.\e_{i+n+1})$ should be a selection functional for the
set~$\prod_{i} K_{i+n+1}$. But a searchable set is exhaustible by
Lemma~\ref{implicit}. To implement the proof of this lemma in our
situation, for any given $p,n,x_n$, define
\begin{eqnarray*}
p_{n,x_n}(\alpha)& = & p(x_0x_1\dots x_{n-1} x_{n}\alpha) \\
& = & p(\Pi(\e)(p)(0) * \Pi(\e)(p)(1) * \operatorname{\dots} * \Pi(\e)(p)(n-1) * x_n * \alpha).
\end{eqnarray*}
For the sake of clarity, here we have used ``$*$'', rather than
juxtaposition as above, to indicate concatenation of
elements and sequences.  Then
\licsmath{\exists \alpha\in \prod_{i} K_{i+n+1}.p(x_0x_1\dots x_n\alpha)}
is equivalent to
\licsmath{p_{n,x_n}(\Pi(\lambda i.\e_{n+i+1})(p_{n,x_n})).}
To find $x_n$ such that this holds, we use $\e_n$:
\licsmath{\Pi(\e)(p)(n) = \e_n(\lambda
  x_n.p_{n,x_n}(\Pi(\lambda i.\e_{n+i+1})(p_{n,x_n}))).}
Because we don't want a different variable $x_n$ for each $n$, we
rename the variable to simply~$x$. This completes our derivation of
the product functional:
\pagebreak[3]
\begin{defi} \label{product:functional} The product functional
  $\Pi \colon ((\D \to \pBool) \to \D)^\myomega \to ((\D^\myomega \to
  \pBool) \to \D^\myomega)$ is recursively defined by
  \licsmath{\Pi(\e)(p)(n) = \e_n(\lambda
    x.p_{n,x,\e}(\Pi(\e^{(n+1)})(p_{n,x,\e})))}
where
\licsmath{
p_{n,x,\e}(\alpha)= p\left(\lambda i.
\begin{cases}
\Pi(\e)(p)(i) & \text{if $i < n$,}\\
x & \text{if $i = n$,} \\
\alpha_{i-n-1} & \text{if $i>n$,}
\end{cases}
\right)\qquad\qquad} and where for any sequence $\beta$ we write
$\beta^{(k)}$ to denote the sequence $\beta$ with the first~$k$
elements removed:
\[
\beta^{(k)}_i = \beta_{k+i}.
\]
\pagebreak[3]

\noindent
For future use, we also write $\beta' = \beta^{(1)}$ and
\[ \overline{\beta}(n) = \langle\beta_0, ... , \beta_{n-1}\rangle \]
so that
\[ p_{n,x,\e}(\alpha) = p(\overline{\Pi(\e)(p)}(n) * x * \alpha).\]
\qed
\end{defi}

The original proof of the following theorem, sketched
in~\cite{escardo:lics07}, uses an auxiliary recurrence relation and
dependent choices. The following more elegant proof, based on
alternative recurrences and bar induction, was presented to me by
Ulrich Berger and is included with his permission:
\begin{thm} \label{searchable:tychonoff} If each $\e_i$
  is a selection functional for a set $K_i \subseteq \D$ then
  $\Pi(\e)$ is a selection functional for the set~$\prod_i
  K_i \subseteq D^\omega$.
\end{thm}
\begin{proof}
Define
\begin{eqnarray*}
p_x(\alpha) & = & p(x * \alpha), \\
x_{\e,p} & = & \Pi(\e)(p)(0) = \e_0(\lambda x.p_x(\Pi(\e')(p_x))), \\
p_\e & = & p_{x_{\e,p}} = p_{0,x_{\e,p},\e}.
\end{eqnarray*}
Claim:
\begin{enumerate}
\item[(1)$(n)$] \quad $p_{n+1,x,\e}  = (p_\e)_{n,x,\e'}$,
\item[(2)$(n)$] \quad $\Pi(\e)(p)(n+1) = \Pi(\e')(p_\e)(n)$,
\item[(3)\phantom{$(n)$}] \quad $\Pi(\e)(p) = x_{\e,p} * \Pi(\e')(p_\e)$.
\end{enumerate}
We prove the properties (1)$(n)$ and (2)$(n)$ simultaneously by course
of values induction.

\noindent
\emph{Proof of (1)$(n)$ assuming $(2)(k)$ for all $k<n$:} 
By the assumption,
\begin{eqnarray*}
\overline{\Pi(\e)(p)}(n+1) 
& = & \Pi(\e)(p)(0) * \langle\Pi(\e)(p)(1), ... , \Pi(\e)(p)(n)\rangle \\
& = & x_{\e,p} * \langle\Pi(\e')(p_\e)(0), ... , \Pi(\e')(p_\e)(n-1)\rangle \\
& = &  x_{\e,p} * \overline{\Pi(\e')(p_\e)}(n).
\end{eqnarray*}
Hence
\begin{eqnarray*}
p_{n+1,x,\e}(\alpha) 
& = & p(\overline{\Pi(\e)(p)}(n+1) * x * \alpha) \\
& = & p(x_{\e,p} * \overline{\Pi(\e')(p_\e)}(n) * x * \alpha) \\
& = & p_\e(\overline{\Pi(\e')(p_\e)}(n) * x * \alpha)\\
& = & (p_\e)_{n,x,\e'}(\alpha).
\end{eqnarray*}

\noindent
\emph{Proof of (2)$(n)$ assuming $(1)(n)$:}
 \begin{eqnarray*}
\Pi(\e)(p)(n+1) 
& = & \e_{n+1}(\lambda x.p_{n+1,x,\e}(\Pi(\e^{(n+2)})(p_{n+1,x,\e}))) \\
& = & \e'_n (\lambda x.(p_\e)_{n,x,\e'}(\Pi(\e'^{(n+1})((p_\e)_{n,x,\e})))  \\
& = & \Pi(\e')(p_\e)(n).
\end{eqnarray*}

\noindent
\emph{Proof of (3):} We have $(\Pi(\e)(p))' = \Pi(\e')(p_\e)$ by (2),
and hence
\[ \Pi(\e)(p) = \Pi(\e)(p)(0) * (\Pi(\e)(p))' = x_{\e,p} *
\Pi(\e')(p_\e).\]
This completes the proof of the claim.

\newcommand{\Solve}{\operatorname{S}} For a subset~$L$ of a domain~$E$
and $q \colon E \to \pBool$ defined on $L$, say that an element $x \in
L$ solves $q$ over $L$ if $q(x) \ne \bot$, and $q(x)=\True$ provided
$q(y) = \True$ for some $y \in L$.  Then $\delta \colon (E \to \pBool)
\to E$ is a selection function for~$L$ iff $\delta(q)$ solves~$q$
over~$L$ for every~$q$ defined on~$L$.  Define the proposition
$\Solve(p,\prod_i K_i)$ by
\begin{eqnarray*}
  \Solve(p,\prod_i K_i) & \iff & \textstyle{\text{$\Pi(\e)(p)$ solves $p$ over $\prod_i
  K_i$ whenever $\e_i$ is a}} \\[-1.5ex]
& & \text{selection functional for~$K_i$.}
\end{eqnarray*}
We need to show that if $p$ is defined on $\prod_i K_i$ then
$\Solve(p,\prod_i K_i)$ holds.  The set of continuous predicates $p$
defined on $\prod_i K_i$ can be defined as follows by bar induction:
\begin{enumerate}
\item if $p(\bot) \ne \bot$ then $p$ is defined on $\prod_i K_i$, and
\item if $p_x$ is defined on $\prod_i K_{i+1}$ for all $x \in K_0$ then
$p$ is defined on $\prod_i K_i$.
\end{enumerate}
Therefore, it suffices to show that for all $p \in (D^\omega \to \pBool)$,
\begin{enumerate}
\item[(i)] if $p(\bot) \ne \bot$ then $\Solve(p,\prod_i K_i)$, and
\item[(ii)] if $\Solve(p_x,\prod_i K_{i+1})$ for all $x \in K_0$ then $\Solve(p,\prod_i K_i)$.
\end{enumerate}
\noindent\emph{Proof of (i):} Let $b=p(\bot)$ and assume that $b \ne\bot$.  Then $p =
\lambda \alpha.b$, by monotonicity of $p$. Let $\e_i$ be a selection
function for~$K_i$. Then $\Pi(\e)(p)(n) = \e_n (\lambda x.b)$.  Since
$\lambda x.b$ is defined on~$K_n$, 
it follows that $\e_n (\lambda x.b) \in K_n$. Hence $\Pi(\e)(p) \in
\prod_i K_i$. If $p(\alpha) = \True$ for some $\alpha \in \prod_i K_i$
then $p = \lambda \alpha.\True$ and hence $p(\Pi(\e)(p)) = \True$.

\medskip
\noindent
\emph{Proof of (ii):} Assume the bar induction hypothesis
\begin{quote}
($\dagger$)  $\Solve(p_x,\prod_i K_{i+1})$ for all $x \in K_0$.
\end{quote}
We need to show that if $\e_i$ is a selection function for~$K_i$ then:
\begin{enumerate}  
\item[(ii)(a)] $\Pi(\e)(p)  \in  \prod_i K_i$.
\item[(ii)(b)] If $p(\alpha) = \True$ for some $\alpha \in \prod_i K_i$, then $p(\Pi(\e)(p)) = \True$.
\end{enumerate}
\noindent
\emph{Proof of (ii)(a):} We show that $\Pi(\e)(p)(n) \in K_n$ by
induction on~$n$.

\medskip
\noindent
\emph{Base case for (ii)(a):} $\Pi(\e)(p)(0) = x_{\e,p} = \e_0
(\lambda x.p_x(\Pi(\e')(p_x)))$.  Since $e'_i$ is a selection
function for $K_{i+1}$, it follows from ($\dagger$) that
$\Pi(\e')(p_x)$ solves $p_x$ over $\prod_i K_{i+1}$ for all $x \in
K_0$. Then $\lambda x.p_x(\Pi(\e')(p_x))$ is defined on $K_0$, and
since $\e_0$ is a selection function for $K_0$, it follows that
$\Pi(\e)(p)(0) \in K_0$.

\medskip
\noindent
\emph{Induction step for (ii)(a):} $\Pi(\e)(p)(n+1) =
\Pi(\e')(p_{\e_n}) = \Pi(\e')(p_{x_{\e,p}})(n) \in K_{n+1} $ by
Claim~(2), by ($\dagger$) and by the fact that $x_{\e,p} = \Pi(\e)(p)(0)
\in K_0$ (base case).

\medskip
\noindent
\emph{Proof of (ii)(b):} Assume $p(\alpha) = \True$ for some $\alpha
\in \prod_i K_i$.  Then $\alpha' \in \prod_i K_{i+1}$ and
$p_{\alpha_0}(\alpha') = \True$, and, by ($\dagger$), we have
$\Solve(p_{\alpha_0},\prod_i K_{i+1})$. Since $\e'$ is a selection
function for $\prod_i K_{i+1}$, it follows that $\Pi(\e')
(p_{\alpha_0}) \in \prod_i K_{i+1}$ and $p_{\alpha_0}(\Pi(\e')
(p_{\alpha_0})) = \True$.  Then $p_\e(\Pi(\e')(p_\e)) = \True$ because
$\e_0$ is a selection function for $K_0$. But $p_\e(\Pi(\e')(p_\e)) =
p(x_{\e,p} * \Pi(\e')(p_\e)) = p(\Pi(\e)(p))$, by Claim~(3).
\end{proof}

\pagebreak[3]
\begin{examples}
  \leavevmode
  \begin{enumerate}
  \item The Cantor space $2^\omega \subseteq \pBool^\omega$ is
    searchable.  A selection functional is given by $\Pi(\lambda
    i.\e_{\Bool})$ where $\e_{\Bool}$ is a selection
    functional for the finite set $\Bool \subseteq \pBool$.
  \item If $K_i \subseteq \N$ is a sequence of finite sets that are
    finitely enumerable uniformly in $i$, then $\prod_i K_i \subseteq
    \pN^\omega$ is searchable, again using the product functional.
  \end{enumerate}
  If a product $\prod_i K_i$ is searchable, then each set $K_n$ is
  searchable uniformly in $n$, by Proposition~\ref{image} as it is the
  computable image of $\prod_i K_i$ under the $n$-th projection.  \qed
\end{examples}

\begin{rem}
  Berger's selection algorithm  $\e 
\colon (\pBool^\myomega \to \pBool) \to \pBool^\myomega$ for
the Cantor space, mentioned in the introduction, can
be written as
\[
\e(p) = 
\begin{cases}
\False  * \e(\lambda \alpha.p(\False * \alpha)) & \text{if $p(\False * \e(\lambda \alpha.p(\False  * \alpha)))$,} \\
\True *  \e(\lambda \alpha.p(\True *  \alpha)) & \text{otherwise.}
\end{cases}
\]
If one defines $\exists \colon (\pBool^\myomega \to \pBool) \to \pBool$ by
$\exists(p) = p(\e(p))$, as in the proof of Lemma~\ref{implicit},
then the above definition is equivalent to
\[
\e(p) = 
\begin{cases}
\False  * \e(\lambda \alpha.p(\False *  \alpha)) & \text{if $\exists \alpha. p(\False * \alpha)$,} \\
\True *  \e(\lambda \alpha.p(\True *  \alpha)) & \text{otherwise.}
\end{cases}
\]
Our product algorithm is inspired by this idea. \qed
\end{rem}

\medskip From now on, we rely on Section~\ref{higher:background} for
the definition of totality.  
By Lemma~\ref{implicit} above and by Theorem~\ref{ex:main} below, a
non-empty set of total elements is exhaustible iff it is searchable,
and hence the above theorem shows that non-empty, exhaustible sets of
total elements are closed under countable products. For
Sierpinski-valued, rather than boolean-valued, universal
quantification functionals, a countable-product algorithm is given
in~\cite{escardo:barbados}, but we don't know how to approach
countable products of boolean-valued quantifiers without the detour
via selection functionals at the time of writing.

We now derive a uniform continuity principle from
Theorem~\ref{searchable:tychonoff}, motivated by topological theorems
that assert that, in certain contexts, continuous functions are
uniformly continuous on compact sets. Define
\[
\text{$\alpha =_n \beta \iff \alpha_i = \beta_i$ for all
    $i<n$,} \qquad\qquad \alpha_{|n}(i)=\begin{cases}
    \alpha_i & \text{$i < n$,}\\
    \bot & \text{otherwise.}\end{cases}
\] 
Then $\alpha =_n \beta$ iff $\alpha_{|n} = \beta_{|n}$.

\begin{thm} \label{uniform:cont} If $f \in (D^\myomega \to \pN)$
  is defined on a product $\prod_i K_i$ of searchable sets, then there
  is a number $n$ such that for all $\alpha,\alpha' \in \prod_i K_i$,
  \,\,\licsmath{\alpha =_n \alpha' \implies f(\alpha)=f(\alpha').}
\end{thm}
\begin{proof}
  Let $(==) \in (\pN \times \pN \to \pBool)$ be the unique
  total function such that
  $(x == y)=\True$ iff $x \sim y$. 
  Then $\forall_{\prod_i {K_i}}(\lambda \alpha.f(\alpha) ==
  f(\alpha))=\True$.  If we define \[ f_{|n}(\alpha)=f(\alpha_{|n}),\]
  then $f=\bigsqcup_n f_{|n}$ and hence $(\lambda \alpha.f(\alpha) ==
  f(\alpha)) = \bigsqcup_n (\lambda \alpha.f_{|n}(\alpha) ==
  f(\alpha))$. So, by continuity of $\forall_{\prod_i {K_i}}$, there
  is $n$ such that \[ \forall_{\prod_i {K_i}}(\lambda
  \alpha.f_{|n}(\alpha) == f(\alpha))=\True.\] We cannot conclude that
  $f_{|n}(\alpha) == f(\alpha)$ for all $\alpha \in \prod_i K_i$
  because there is no reason why the predicate $\lambda
  \alpha.f_{|n}(\alpha) == f(\alpha)$ should be defined on~$\prod_i
  K_i$.  To overcome this difficulty, let $\beta \in \prod_{i} K_{i +
    n}$ and define
  $g_n(\alpha)=f(\alpha_0\alpha_1\dots\alpha_{n-1}\beta)$ so that
  $g_n$ is defined on~$\prod_i K_i$ and
  above~$f_n$. 
  By monotonicity, $ \forall_{\prod_i {K_i}}(\lambda
  \alpha.g_n(\alpha) == f(\alpha))=\True.  $ Now the predicate
  $\lambda \alpha.g_n(\alpha) == f(\alpha)$ is defined on $\prod_i
  K_i$ and hence $g_n(\alpha) = f(\alpha)$ for all $\alpha \in \prod_i
  K_i$.  But if $\alpha =_n \alpha'$ then $g_n(\alpha)=g_n(\alpha')$,
  and so $f(\alpha)=f(\alpha')$, as required.
\end{proof}

\medskip
The following is an immediate consequence of this and
Theorem~\ref{searchable:tychonoff}:
\begin{cor} \label{fan}
  The functional $\operatorname{fan}=\operatorname{fan}_{\prod_i K_i}
  \colon (\D^\myomega \to \pN) \to \pN$ defined by
\[
  \operatorname{fan}(f) = \mu n.\forall \alpha,\beta \in \prod_i K_i.\,\, \alpha =_n \beta \implies f(\alpha)=f(\beta)
\]
is computable uniformly in any sequence of selection functionals for
the sets~$K_i \subseteq D$, and is defined on any $f$ that is defined
on $\prod_ i {K_i}$. Moreover, if the sets $K_i$ consist of total
elements of a domain $D=D_\sigma$, then the fan functional is total.
\end{cor}
This holds, in particular, if $\D=\pN$ and each $K_i$ is a finite
subset of $\N$ defined uniformly in $i$, which is the case that has
been considered in higher-type computability theory regarding the fan
functional (see e.g.~\cite{gandy77}).
Here we have generalized this to
arbitrary higher types $\D=\D_{\sigma}$.
A consequence of the exhaustibility of the Cantor space is that:
\begin{cor}\label{decidable:equivalence}
  The total elements of the function space $(\pBool^\omega \to \pN)$
  have decidable equivalence.
\end{cor}
\begin{proof}
  The algorithm $(==) \colon (\pBool^\omega \to \pN) \times (\pBool^\omega
  \to \pN) \to \pBool$ given by  \[(f == g) = \forall \alpha \in
  \Bool^\omega. f(\alpha) == g(\alpha)\]
  does the job.
\end{proof}

This can be generalized as follows, where we now rely on
Section~\ref{kk:background} for the definition of the Kleene--Kreisel
spaces~$C_\sigma$.  
\begin{defi} \label{discrete:compact}
The \emph{discrete} and \emph{compact}
types are inductively defined as
\newcommand{\discrete}{\mathtt{discrete}}
\newcommand{\compact}{\mathtt{compact}}
\begin{eqnarray*}
\discrete & ::= & o \mid \iota \mid \discrete \times \discrete \mid \compact \to \discrete, \\
\compact & ::= & o \mid \compact \times \compact \mid \discrete \to \compact.
\end{eqnarray*}
The reason for this terminology is that the space $C_{\sigma}$ is
discrete if $\sigma$ is discrete, and it is compact if $\sigma$ is
compact, as observed in~\cite{escardo:barbados}. \qed
\end{defi}
\pagebreak[3]
\begin{thm} \label{thm:discrete:compact}
\leavevmode
\begin{enumerate}
\item If $\sigma$ is discrete, then $C_{\sigma}$ is computably enumerable. 
\item The total elements of a domain of compact type form a searchable set.
\item The total elements of a domain of discrete type have decidable equivalence.
\end{enumerate}
\end{thm}
\begin{proof}
  By induction on the definitions of discrete and compact type.  The
  first condition holds by the Kleene--Kreisel density theorem, which
  gives a computable dense sequence of~$C_\sigma$, and by the fact
  that $C_\sigma$ is discrete.  For the second condition, use
  Theorem~\ref{searchable:tychonoff} with the aid of the first
  condition, and, for the third one, use the argument of
  Corollary~\ref{decidable:equivalence}.
\end{proof}

We conclude this section with a natural notion that plays a
fundamental role in our investigation of exhaustible and searchable
sets and their relationship. Let $D=D_{\sigma}$ and $D'=D_{\sigma'}$
for types~$\sigma$ and~$\sigma'$.
\begin{defi} \label{def:entire} We say that a set $K \subseteq
  D$ is \emph{entire} if it consists of total elements and is
  closed under total equivalence. \qed
\end{defi}
Notice that if $p$ is total then it is defined on every entire set.
If $p$ is not total and $K$ is not entire, but if $p$ is defined
on~$K$, then $p(x) = p(x')$ for all $x \sim x'$ in $K$, because if $x
\sim x'$ then $x$ and $x'$ are bounded above and hence so are $p(x)$
and $p(x')$, which then must be equal as they are non-bottom by
definition.  But if $x \in K$ and $x' \sim x$ for $x'$ outside $K$, it
doesn't follow that $p(x') \ne \bot$ (consider e.g.\ $K=\{\lambda
i.\True\}$ for $\sigma=\iota \to o$ and $p(\alpha)=\alpha(\bot)$).

\medskip
The following closure properties of entire sets are easily verified:
\begin{enumerate}
\item If $K \subseteq \D$ and $K' \subseteq \D'$ are entire, so is $K
  \times K' \subseteq \D \times \D'$
\item If $K_i$ is a sequence of entire subsets of $\D$, then
$\prod_i K_i$ is an entire subset of $\D^\myomega$.
\end{enumerate}
\begin{defi}
  The image of an entire set by a total function doesn't need to be
  entire, but it consists of total elements, and hence its closure
  under total equivalence is entire. We refer to this as its
  \emph{entire image}. (Thus, entire images are defined for total
  functions and entire sets only.) \qed
\end{defi}

\begin{prop} \label{image:bis} Exhaustible and searchable sets are
  closed under the formation of computable entire images.
\end{prop}
\begin{proof}
For given $f \colon \D \to \D'$ and $K \subseteq \D$ exhaustible,
consider the quantification functional
\licsmath{\forall_{f(K)}(q)=\forall x \in K.q(f(x)).}
defined in the proof of Proposition~\ref{image}.
If $f$ is total and $K$ is
entire with entire image $L$, then we can take
$\forall_{L}=\forall_{f(K)}$. To verify this, let $q$ be defined on
$L$. Then $q$ is defined on $f(K) \subseteq L$, and hence if
$q(l)=\True$ for all $l \in L$, then $\forall_{L}(q)=\True$. If, on the
other hand, $q(l)=\False$ for some $l \in L$, then $l \sim f(x)$ for
some $x \in K$. But then $q(f(x))=\False$, and so
$\forall_{T}(q)=\False$, which concludes the verification.
The argument for searchable sets is similar. 
\end{proof}

\begin{defi} \label{semidecidable} \leavevmode Let $\Sierp = \{
  \bot, \top\}$ by the Sierpinski domain and $F \subseteq D=D_\sigma$
  be entire.
  \begin{enumerate}
  \item $F$ is \emph{decidable} if there is
    a total computable map $\psi_F \colon D \to \pBool$ such that, for
    all total $x \in D$, $\psi_F(x)=\True$ iff $x \in F$.

  \item $F$ is \emph{semi-decidable} if
    there is a computable map $\chi_F \colon D \to \Sierp$ such that,
    for all total $x \in D$, $\psi_F(x)=\top$ iff $x \in F$.

  \item $F$ is \emph{co-semi-decidable} if its complement in
    $T=T_\sigma$ semi-decidable. \qed
  \end{enumerate}
\end{defi}
Notice that the functions $\psi_F$ is not uniquely determined by $F$,
because its behaviour is specified on a subset of~$D$, but that
$\chi_F$ is uniquely determined by~$F$. Notice also 
$F$ is decidable if and only if it is decidable on~$T$ in the sense of
Definition~\ref{decidableonk} with $K=T$.

\section{Compactness of exhaustible sets} \label{criteria}

A notion analogous to exhaustibility, with the Sierpinski
domain~$\Sierp$ playing the role of the boolean
domain~$\pBool$, is considered in~\cite{escardo:barbados}.  A crucial
fact, formulated here as Lemma~\ref{crucial:generalized}, is that the
(now unique) quantification functional $\forall_K \colon (\D \to \Sierp)
\to \Sierp$ is continuous iff the set~$K$ is compact in the Scott
topology of~$\D$.  Hence, because computable functionals are
continuous, Sierpinski-exhaustible sets are compact, and so Sierpinski
exhaustibility is seen as articulating an algorithmic version of the
topological notion of compactness. The computational idea is that,
given any \emph{semi-decidable} property of $\D$, one can
\emph{semi-decide} whether it holds for all elements of~$K$. Closure
properties analogous to the above are established for Sierpinski
exhaustibility in~\cite{escardo:barbados}. 

The present investigation can be seen as a natural follow-up of that
work that arises by asking what changes if one moves from
semi-decision problems to decision problems.  One significant change
is that continuity of a quantification functional $\forall_K \colon
(\D \to \pBool) \to \pBool$ doesn't entail the compactness of $K$ in
the Scott topology any longer:

\pagebreak[3]
\begin{examples} \leavevmode
\label{counter:example} 
\begin{enumerate}
\item \label{counter:example:1} \emph{There are exhaustible sets that
    fail to be compact in the Scott topology.}

  \medskip
  \noindent
  By~\cite{scott:datatypes,plotkin:tomega}, any second-countable
  $T_0$~space, e.g.\ the real line~$\R$, can be embedded into the
  domain $\D=\pBool^\myomega$ under the Scott topology. But $\R$ is a
  connected space, which is equivalent to saying that every continuous
  boolean-valued map defined on it is constant. Hence a predicate $p
  \in (\D \to \pBool)$ is defined on $\R$ iff it is constant on~$\R$.
  Therefore $\R$ is trivially exhaustible: $\forall_{\R}(p)=p(0)$.
  But it is not compact.

  Notice also that any space embedded into the total elements of
  $\pBool^\myomega$ must be totally disconnected, and hence any
  embedding of $\R$ into $\pBool^\myomega$ must assign non-total
  elements of $\pBool^\myomega$ to some real numbers. One may suspect
  that if such embeddings are ruled out, this problem would disappear.
  But this is not the case, as the next example shows.

\item \label{counter:example:2} \emph{There are exhaustible sets of
    total elements that fail to be Scott compact.}

  \medskip
  \noindent
  In fact, there is a trivial and pervasive counter-example. Let $f
  \in ((\pN \to \pN) \to \pN)$ be total.  Then the total equivalence
  class $K$ of $f$, as is well known and easy to verify, doesn't have
  minimal elements, and hence cannot be compact in the Scott topology.
  But it is exhaustible with $\forall_K(p)=p(f)$.  \qed
\end{enumerate}
\end{examples}

One may feel somewhat cheated by the second counter-example, because
although the set $K$ is not Scott compact, it is generated by the
singleton $\{f\}$, which is Scott compact, and because we took
$\forall_K$ to be~$\forall_{\{f\}}$ (cf.\ the proof of
Proposition~\ref{image:bis}).
Lemma~\ref{lemma:criterion}(\ref{lemma:criterion:3}) below shows that
any counter-example is generated by a Scott compact set in a similar
fashion.  In any case, although exhaustible sets do fail to be compact
in the Scott topology, if they consist of total elements then they are
compact in the Kleene--Kreisel topology.  In order to formulate and
prove this, we need some definitions. We now rely on
Section~\ref{kk:background}.
\begin{defi} 
  Let $\sigma$ be a type, $\D=\domain{\sigma}$, $\T=\total{\sigma}$,
  $\C=\kk{\sigma}$ and $\qq = \quo{\sigma} \colon \total{\sigma} \to
  \kk{\sigma}$. 
\begin{enumerate}
\item By the \emph{shadow} of a set $K \subseteq \T$ we mean its
  $\qq$-image in~$\C$. 
  Similarly, by the shadow of an element $x \in T$ we mean its
  $\qq$-image $\qq(x)$ in $\C$.

\item A set $K \subseteq T$ is called \emph{Kleene--Kreisel compact}
  if its shadow is compact. \qed
\end{enumerate}
\end{defi}
Recall that the \emph{Cantor space} is the set $\Bool^\omega$ of
maximal elements of $\pBool^\omega$.
\begin{rem}
  Sometimes, for example for the implementation of the product
  functional defined in Section~\ref{building} in the language PCF,
  which lacks countable powers, one works with the Cantor space within
  the function space $\pBool^\pN$. The Cantor space is homeomorphic to
  the subspace of total strict functions $\alpha \in \pBool^\pN$,
  where $\alpha$ is strict if $\alpha(\bot)=\bot$.  It is also
  homeomorphic to the quotient of the set of \emph{all} total elements
  of~$\pBool^\pN$. But notice that the set of \emph{maximal} elements
  of $\pBool^\pN$ is \emph{not} homeomorphic to the Cantor
  space.  This is because the two non-strict elements $\lambda
  i.\False$ and $\lambda i.\True$ are finite (or order compact), and
  hence isolated in the relative Scott topology (meaning that the two
  corresponding singletons are open), and hence the maximal elements
  have a topology strictly finer than that of the Cantor space, as
  there are no isolated points in the Cantor space. As is well known
  in topology, no compact Hausdorff topology can have another compact
  Hausdorff topology as a strict refinement. \qed
\end{rem}

Every (computationally) exhaustible set is topologically exhaustible
in the sense of the following definition, because computable maps are
continuous.
\begin{defi} \label{topologically:exhaustible} \leavevmode
We say that a set $K \subseteq D$ is \emph{topologically
    exhaustible} if there is a continuous map $\forall_K \in ((D \to
  \pBool) \to \pBool)$ satisfying the conditions of
  Definition~\ref{exhaustible:def}. \qed
\end{defi}
\noindent The following is our main tool in the constructions and
proofs of correctness and termination of algorithms developed in
Sections~\ref{characterization}--\ref{arzela:ascoli}.  Its proof relies on
Sections~\ref{compactly} and~\ref{hyland}.  \pagebreak[3]
\begin{lem} \label{lemma:criterion} \leavevmode 
  \begin{enumerate}
  \item \label{lemma:criterion:1} Any topologically
    exhaustible set of total elements is Kleene--Kreisel compact.

  \item \label{lemma:criterion:2} Any non-empty, Kleene--Kreisel
    compact entire set is an entire continuous image of the Cantor
    space and hence is topologically exhaustible.


  \item \label{lemma:criterion:3} Any Kleene--Kreisel compact entire
    set has a Scott compact subset with the same shadow.
  \end{enumerate}
\end{lem}
\pagebreak[3]
\begin{proof}
  (\ref{lemma:criterion:1}): Let $K \subseteq \T$ be exhaustible. By
  Lemma~\ref{zero:compact} and the fact that clopen sets are closed
  under finite unions, to establish compactness of $\rho(K)$, it is
  enough to consider a directed clopen cover $\mathcal{U}$. By
  Lemma~\ref{clopen:extension}, for every $U \in \mathcal{U}$ there is
  a total $p_U \in (\D \to \pBool)$ with
  \licsmath{(\dagger) \quad \text{$\qq^{-1}(U) \subseteq
      p_U^{-1}(\True)$ and $\qq^{-1}(\C \setminus U) \subseteq
      p_U^{-1}(\False)$.}}
  Define predicates $q_U, r \in (\D \to \pBool)$ by
  \[ q_U^{-1}(\True)=p_U^{-1}(\True), \quad
    r^{-1}(\True)=\bigcup_{U \in \mathcal{U}} p_U^{-1}(\True), \quad
    q_U^{-1}(\False)=r^{-1}(\False)=\emptyset.
  \]  
  Then $q_U \sqsubseteq p_U$, the set $\{ q_U \mid U \in
  \mathcal{U}\}$ is directed, and $r=\bigsqcup_{U \in \mathcal{U}}
  q_U$. Because $\rho(K) \subseteq \bigcup \mathcal{U}$, we have that
  $K \subseteq r^{-1}(\True)$ and hence $\forall_K(r)=\True$. So, by
  continuity of $\forall_K$, there is $U \in \mathcal{U}$ with
  $\forall_K(q_U)=\True$, and hence with $\forall_K(p_U)=\True$ by
  monotonicity.  Let $x \in K$. Then $p_U(x)=\True$ by specification
  of~$\forall_K$ and the fact that $p_U$ is total and hence defined on
  $K$.  But then $\qq(x) \in U$, for otherwise $(\dagger)$ would
  entail $p_U(x)=\False$.  This shows that $\qq(K) \subseteq U$, and
  so $\qq(K)$ is compact.

\medskip

(\ref{lemma:criterion:2}): By e.g.\ \cite{escardo:lawson:simpson}, any
compact subset of $C$ is countably based (even though $C$ is not). But
any non-empty compact Hausdorff countably based space is a continuous
image of the Cantor space. Hence there is a continuous map
$\Bool^\omega \to C$ with image $\rho(K)$ for any entire set $K
\subseteq D$. Then the entire image of the Cantor space under any
representative $\pBool^\myomega \to D$ is~$K$.

(\ref{lemma:criterion:3}): This follows from the argument given in
(\ref{lemma:criterion:2}), because the Cantor space is Scott compact.
\end{proof}

%
\begin{rem} \label{introduction}
In particular, this gives a topological view of the computational fact
stated in the introduction that exhaustible sets of natural numbers
must be finite: all compact sets are finite in a discrete space. \qed
\end{rem}
Kleene--Kreisel compactness can be expressed as a finite-subcover
condition for the Scott topology as follows: An entire set $K
\subseteq D$ is Kleene--Kreisel compact if and only if every cover of
$K$ by Scott open sets that are closed under total equivalence has a
finite subcover.  This is a straightforward consequence of the fact
that the Kleene--Kreisel topology is the quotient by total equivalence
of the relative Scott topology on the total elements.




We also remark that there is a natural topology on $D$, coarser than
the Scott topology, in which \emph{all} exhaustible sets are compact.
Part of the argument of
Lemma~\ref{lemma:criterion}(\ref{lemma:criterion:1}) shows that any
exhaustible set $K$ is compact in the coarsest topology on $D$ such
that all predicates $p \in (D \in \pBool)$ defined on $K$ are
continuous.  This is generated by directed unions of basic open sets
of the form $p^{-1}(\True)$ with $p$ as above, because such sets are
closed under finite unions and intersections.  This construction is
analogous to the zero-dimensional reflection of a topology, and
happens to coincide with it in the case considered in
Lemma~\ref{lemma:criterion}(\ref{lemma:criterion:1}), modulo
quotienting, and can also be compared with the weak topology in
functional analysis.

\section{Searchability  of exhaustible sets} 
\label{characterization}

We already know that every searchable set is exhaustible
(Lemma~\ref{implicit}).  This implication is uniform, in the sense
that there is a computable functional \[ ((\D \to \pBool) \to D) \to
((\D \to \pBool) \to \pBool)\] that transforms selection functionals
into quantification functionals, namely $\e \mapsto \lambda
p.p(\e(p))$. We now establish the converse for non-empty
entire sets, and some additional results. The fact that exhaustible
entire sets are Kleene--Kreisel compact, established in the previous
section, plays a fundamental role in the construction of the
algorithms
\begin{defi} \label{total:retract} We say that a set $S
  \subseteq \D = \D_{\sigma}$ is a \emph{retract up to total
    equivalence} if there is a function $r \in (D \to D)$ such that
\begin{enumerate}
\item $r(x) \in S$ for all total $x \in \D$,
\item $r(s) \sim s$ for all $s \in S$. \qed
\end{enumerate}
\end{defi}
In this case, $r$ is total, all elements of $S$ are total, and
$r(r(x)) \sim r(x)$ for all total~$x$. Notice that $r$ is
a retract up to total equivalence iff it is a total function and its
Kleene--Kreisel shadow $\rho(r) \colon C \to C$ is a retract in the
usual topological sense, where $C=C_{\sigma}$.
\begin{defi}
  We say that two entire sets $K \subseteq D = D_\sigma$ and $L
  \subseteq E = D_\tau$ are \emph{homeomorphic up to total
    equivalence} if there are total functions $f \in (D \to E)$ and $g
  \in (E \to D)$ such that $g(f(x)) \sim x$ and $f(g(y)) \sim y$ for
  all $x \in K$ and $y \in L$. \qed
\end{defi}
This is equivalent to saying that the shadow functions $\rho(f)$ and
$\rho(g)$ restrict to a (true) homeomorphism between the shadows
of~$K$ and~$L$. In this case, $L$ is the entire $f$-image of~$K$, and
$K$ is the entire $g$-image of~$L$.

\begin{thm} \label{ex:main} If $K \subseteq D=D_{\sigma}$ is a
  non-empty, exhaustible entire set then, uniformly in any
  quantification functional for $K$:
  \begin{enumerate}
  \item $K$ is searchable.
  \item $K$ is a computable entire image of the Cantor space.
  \item $K$ is computably homeomorphic to some entire exhaustible
    subset of the Baire domain~$\pN^\myomega$, up to total equivalence.
  \item $K$ is a computable retract up to total equivalence.
  \item  $K$ is co-semi-decidable.
  \end{enumerate}  
\end{thm}
In particular, after the theorem is proved, one can w.l.o.g.\ work
with total predicates rather than predicates defined on~$K$, as for any
predicate $p \in (D_{\sigma} \to \pBool)$ defined on~$K$ one can
uniformly find a total predicate that agrees with $p$ on $K$, by
composition with the retraction.

\pagebreak[3] 
\begin{proof}
  We proceed by cases, of increasing generality, on the type of~$K$.
  The case $K \subseteq \pN$ is trivial and is implicitly used in the
  case $K \subseteq \pN^\myomega$, which in turn is used in the next
  case $K \subseteq (D \to \pN)$. The general case $K \subseteq D$ is
  reduced to this last case via retracts using Lemma~\ref{pure}.

\medskip

\emph{(i)} Case $K \subseteq \pN$: We can define
\[
\e_K(p) = 
\begin{cases}
  \mu n.\exists m \in K. n = m \wedge p(n) & \text{if $\exists n \in K.p(n),$} \\
  \mu n.\exists m \in K. n = m & \text{otherwise.}
\end{cases}
\]
This construction defines $\e_K$ uniformly in
$\exists_K$. We could now easily show that $K$ satisfies the other
conditions of the theorem, but this won't be required for our proof,
as this will follow in later cases.

\medskip

For future use, notice that if $K \subseteq \pN$ is entire and
exhaustible, then the supremum of the finite set~$K$ (which is zero
if $K$ is empty and the largest element of $K$ otherwise) can be
computed uniformly in any quantification functional for~$K$ as
\[
\sup K  = \mu m. \forall n \in K. n \le m.  
\] 
Hence, the finite enumeration $e_n$ of
the elements of $K$, in ascending order, 
for $0 \le n < \operatorname{cardinality}(K)$, is
uniformly computable as
\[
e_n = \mu y.\exists m \in K. \forall i<n. m \ne e_i \wedge m = y. 
\]
We stop when we find $n$ such that $e_n = \sup K$, and we include
$e_n$ iff $\exists m \in K. m = \sup K$.

\medskip

\emph{(ii)}  Case $K \subseteq \pN^\myomega$.
We first argue that we can find some $\alpha \in K$, uniformly in
$\exists_K$, by the following algorithm defined by course-of-values
induction on $n$:
\[
\alpha_n = \mu k.\exists \beta \in K.  \alpha =_n \beta \wedge \beta_n = k.
\]
Recall that we defined 
\[
\alpha =_n \beta \iff \forall 
i<n.\beta_i = \alpha_i
\]
in the paragraph preceding Theorem~\ref{uniform:cont}.  By
construction, for every $n$ there is $\beta \in K$ with $\alpha =_n
\beta$, and in particular $\alpha$ is total.  Because the shadow of
$K$ is compact, it is closed, and because $K$ is entire, $\alpha \in
K$, as required.  Then we can define, using
Proposition~\ref{prop:intersec} and the above algorithm to construct
$\alpha$ in both cases,
\[
\e_K(p) =
\begin{cases}
  \text{some $\alpha \in K \cap p^{-1}(\True)$} & \text{if $\exists \alpha \in K.p(\alpha)$,}\\
  \text{some $\alpha \in K$} & \text{otherwise.}
\end{cases}
\]
Again,
this construction defines $\e_K$ uniformly in $\exists_K$.

\medskip

We now show that $K$ is a computable retract up to total equivalence,
uniformly in any quantification functional for~$K$.  Define \[ r = r_K
\colon \pN^\myomega \to \pN^\myomega\] by course-of-values induction
on $n$:
\[
 r(\alpha)(n) = 
 \begin{cases}
   \alpha_n & \text{if $\exists \beta \in K. \beta =_n r(\alpha) \wedge \beta_n = \alpha_n$,} \\
   \mu m. \exists \beta \in K. \beta =_n r(\alpha) \wedge
     \beta_n = m & \text{otherwise.}
 \end{cases}
\]
Because the shadow of $K$ is closed, the finite prefixes of its
members form a tree whose infinite paths correspond to the elements
of~$K$. 
The above algorithm follows the infinite path $\alpha$ through the
tree, either for ever (always following the first case) or until the
path exits the tree (reaching the second case). If and when $\alpha$
exits the tree, we replace the remainder of $\alpha$ by the left-most
infinite branch of the subtree at which $\alpha$ exits the tree. Then
$r$ clearly satisfies the required conditions.  

\medskip

A semi-decision procedure for the complement of~$K$ is given by
\[
\alpha \not\in K \iff r(\alpha) \ne \alpha,
\]
using the fact that apartness of total elements of $\pN^\myomega$ is
semi-decidable.
(This is a computational version of the topological fact
that retracts of Hausdorff spaces are closed.) 

\medskip

We now show that $K$ is a computable entire image of the Cantor space.
For any $i$, the set $K_i = \{ \alpha_i \mid \alpha \in K \}$ is
exhaustible by Proposition~\ref{image} as evaluation at $i$ is
computable. It is enough to show that $\prod_i
\{0,1,\dots,\sup K_i\} \subseteq \pN^\myomega$ is an entire image of the
Cantor space by a computable map $t \colon \pBool^\myomega \to
\pN^\myomega$, because then $r \comp t$ has $K$ as its entire image
since $K$ is contained in that product. But this is straightforward:
at each stage $j$ of the computation of $t(\alpha)$, look at the next
$\lceil \log_2(\sup K_j) \rceil$ digits of the input~$\alpha$,
compute the natural
number $f(j)$ represented by this finite sequence, and let
$t(\alpha)(j)=\min(\sup K_j,f(j))$.

\medskip

\pagebreak[3] \emph{(iii)} Case $K \subseteq (D \to \pN)$ where
$D=D_{\sigma}$ for an arbitrary type~$\sigma$: In order to reduce this
to case~\emph{(ii)}, we invoke the Kleene--Kreisel density
theorem, 
to get a computable sequence $d \in D^\myomega$ such that the shadow
sequence $\rho(d_n)$ is dense in $C=\kk{\sigma}$.  Define
\[
  P \colon (D \to \pN) \to \pN^\myomega \\
\]
\[
  P(f) = \lambda n.f(d_n).
\]
We will define a total function $E=E_K$ in the other direction,
\[
  E \colon \pN^\myomega \to  (D \to \pN)
\]
such that
\[
R=E \comp P\colon (\D \to \pN) \to (\D \to \pN)
\]
exhibits $K$ as a retract of $(\D \to \pN)$ up to total equivalence.
Thus, for $f \in K$, one can recover the behaviour of $f$
at total elements from its behaviour on the dense sequence~$d$.
Because this implies that $K$ is the entire image of $P(K)$, and
because $P(K)$ is searchable by case~\emph{(ii)}, it will follow that
$K$ is searchable and an entire image of the Cantor space, because $E$
preserves total equivalence on~$P(K)$.

For $\alpha \in \pN^\myomega$ and $n \in \N$, define
\[
F_n^\alpha = \{ f \in (D \to \pN) \mid \forall i<n. f(d_i) = \alpha_i \}, \qquad\quad
K_n^\alpha = K \cap F_n^\alpha.
\]
Here we regard $\alpha$ as potentially coding the action of some $f$
on the set of elements $d_i$.  The set $F_n^\alpha$ is decidable on
$K$ uniformly in $\alpha$ and $n$, and hence $K_n^\alpha$ is uniformly
exhaustible by Proposition~\ref{prop:intersec}.  \pagebreak[3]
\begin{lem} \label{c} If $K \subseteq (D \to \pN)$ is a
  Kleene--Kreisel compact entire set, then for all total $\alpha \in
  \pN^\myomega$ and all total $x \in D$ there is $n$ such that $f(x) =
  f'(x)$ for all $f,f' \in K_n^\alpha$.
\end{lem}
\begin{proof}
  For any $\g \in \kk{\sigma \to \iota}$ the set $B_{\g} = \{ \f \in
  \kk{\sigma \to \iota} \mid \f(\x) = \g(\x) \}$ is clopen, where we
  write $\x = \qq(x)$. By density of $d_n$, the set $\bigcap_n
  \sK_n^\alpha$ has at most one element, where $\sK_n^\alpha$ denotes
  the shadow of $K_n^\alpha$. Hence if $\g \in \bigcap_n \sK_n^\alpha$
  then $\bigcap_n \sK_n^\alpha = \{ \g \} \subseteq B_{\g}$.  Because
  $\kk{\sigma \to \iota}$ is Hausdorff and because each $\sK_n^\alpha$
  is compact and $B_{\g}$ is open, there is $n$ such that already
  $\sK_n^\alpha \subseteq B_{\g}$. So for all $\f \in \sK_n^\alpha$
  one has $\f(\x) = \g(\x)$, and hence for all $\f,\f' \in
  \sK_n^\alpha$ one has $\f(\x) = \f'(\x)$.
\end{proof}
By Proposition~\ref{image:bis}, the entire $P$-image $L \subseteq
\pN^\myomega$ of $K$ is exhaustible.  Let $r = r_{L}$
be defined as in case~\emph{(ii)}, and define $E \colon \pN^\myomega \to
(D \to \pN)$ by
\[
E(\alpha)(x) = 
  \mu y. \exists f \in K_n^{r(\alpha)}.f(x)=y,
\]
where
  $n$ is the least number such that $\forall f,f' \in
  K_n^{r(\alpha)}.f(x)=f'(x).$
  By exhaustibility of~$K_n^{r(\alpha)}$, this can be found uniformly
  in $\alpha$, and hence $E$ is computable uniformly in~$K$.

\pagebreak[3]
\begin{proof}[Proof of correctness of $E$.]
~

\emph{(a)~$E$ is total and maps $L$ into $K$.} 
Let $\alpha \in \pN^\omega$ be total.  Then $r(\alpha) \in L$, by
construction of $r$, and hence there is $g \in K$ with $r(\alpha) \sim
P(g)$, and so with $g \in K_n^{r(\alpha)}$ for any~$n$. Let $x \in D$
be total and $n$ be the least number such that $f(x)=f'(x)$ for all
$f,f' \in K_n^{r(\alpha)}$. Then $f(x)=g(x)$ for all $f \in
K_n^{r(\alpha)}$, and hence $E(\alpha)(x)=g(x)$.  Therefore $E(\alpha)
\sim g \in K$, and hence $E(\alpha) \in K$ as $K$ is entire, and in
particular $E$ is total.  By construction $E \sim E \comp r$, and
hence, because $r$ exhibits $L$ as a retract up to total equivalence
and $K$ is entire, the $E$-image of $L$ is $K$.

\medskip
\emph{(b)~If $f \in (D \to \pN)$ is total then $R(f)=E(P(f)) \in
    K$.} 
Because $P(f) \in L$.

\medskip
\emph{(c)~If $f \in K$ then $R(f) \sim f$.} 
Continuing from the proof of (a), for $\alpha=P(f)$ we have $r(\alpha)
\sim \alpha$ by construction of $r$, and hence for any $g \in K$ such
that $P(g) = r(\alpha)$ we have $g(d_i) = P(g)(i)=r(\alpha)(i) =
\alpha_i = P(f)(i) = f(d_i)$ and so $g \sim f$ by density, which shows
that $R(f)=E(P(f)) \sim f$, as required.
\end{proof}

\medskip

A semi-decision procedure for the complement of~$K$ is
given as in case \emph{(ii)},
\[
f \not\in K \iff R(f) \ne f,
\]
because 
$
f \ne f' \iff \exists n \in \N. f(d_n) \ne f'(d_n),
$
for total functions $f'$ and $f$ since $K$ is entire and $d$ is dense.

Because $E$ and $P$ are total, they induce computable Kleene--Kreisel
functionals $\sE = \rho(E) \colon \N^\N \to \N^{C}$ and $\sP = \rho(P)
\colon \N^{C} \to \N^\N$ where $C=C_{\sigma}$. If $\sK \subseteq
\N^{C}$ is the shadow of~$K$, then the restriction of $\sP$ to $\sK$
followed by the co-restriction to its image is a homeomorphism $\sK
\to \sP(\sK)$: abstractly because any continuous bijection of compact
Hausdorff spaces is a homeomorphism, and concretely because the
bi-restriction of $\sE$ is a continuous inverse.  Hence any
exhaustible subset of $(D \to \pN)$ is computably homeomorphic to the
shadow of some exhaustible subset of the Baire domain~$\pN^\myomega$,
up to total equivalence.

\medskip

\pagebreak[3] \emph{(iv)} General case. We derive this from the case
\emph{(iii)}. By Lemma~\ref{pure}, for any $D=D_{\sigma}$ there are
$D'=D_{\tau \to \iota}$ and computable $P \colon D' \to D$ and $E
\colon D \to D'$ such $R = E \comp P$ is a retraction up to total
equivalence and $T_{\sigma}$ is the entire image of $P$.  Let $K
\subseteq D$ be a non-empty, exhaustible entire set, and let $K'$ be
the entire $E$-image of $K$. Then $K$ is the entire image of $P(K')$,
and, because $K$ is entire, a predicate $p' \in (D' \to \pBool)$
defined on $K'$ holds for all $x' \in K'$ if and only if $p' \comp E$
holds for all $x \in K$.  Hence $K'$ is exhaustible with
$\forall_{K'}(p')=\forall_K(p'\comp E)$. By case \emph{(iii)} above,
$K'$ is searchable. Therefore $K$ is searchable by
Proposition~\ref{image}. Similarly, the other properties we need to
establish are closed under the formation of retracts and hence are
inherited from case~\emph{(iii)}.  

\medskip This concludes the proof of Theorem~\ref{ex:main}.
\end{proof}

\section{Ascoli--Arzela type characterizations of exhaustible sets} \label{arzela:ascoli}

We reformulate a theorem of Gale's~\cite{gale} that characterizes
compact subsets of function spaces (Theorem~\ref{gale:modified}).
This suggests a characterization of exhaustible entire sets
(Theorem~\ref{exhaustible:arzela}), whose topological version is
developed first (Theorem~\ref{exhaustible:arzela:topological}).  The
main idea is to replace a condition in Gale's theorem by a continuity
condition (Section~\ref{arzela:compact}), and then further replace it
by a computability condition (Section~\ref{arzela}). This method of
transforming topological theorems into computational theorems is the
main thrust of the paper~\cite{escardo:barbados}, which develops many
instances of computational manifestations of topological theorems.

\subsection{Topological version}
\label{arzela:compact}

The Heine--Borel theorem characterizes the compact subsets of
Euclidean space~$\R^n$ as those that are closed and bounded. The
Arzela--Ascoli theorem generalizes this to subsets of $\R^X$, where
$X$ is a compact metric space and $\R^X$ is the set of continuous
functions endowed with the metric defined by \[ d(f,g)=\max \{
d(f(x),g(x)) \mid x \in X\}.\] A set $K \subseteq \R^X$ is compact if
and only if it is closed, bounded and equi-continuous.
Equi-continuity of $K$ means that the functions $f \in K$ are
simultaneously continuous, in the sense that for every $x \in X$ and
every $\epsilon>0$, there is $\delta>0$ such that $d(x,x') < \delta
\implies d(f(x),f(x'))$ for all $x' \in X$ and all $f \in K$. The
Heine--Borel theorem is the particular case in which $X$ is the
discrete space $\{1,\dots,n\}$, for equi-continuity holds
automatically for any subset of $\R^X$ in this case.  The above metric
on $\R^X$ induces the compact-open topology. More general
Arzela--Ascoli type theorems characterize compact subsets of
spaces~$Y^X$ of continuous functions under the compact-open topology,
for a variety of spaces~$X$ and~$Y$, with a number of generalizations
or versions of the notion of equi-continuity, notably even continuity
in the sense of Kelley~\cite{kelley}.

Among a multitude of generalizations of the Arzela--Ascoli theorem,
that of Gale~\cite[Theorem 1]{gale} proves to be relevant concerning
exhaustibility of entire sets:
\pagebreak[3]
\begin{quote}
\em
  If $X$ and $Y$ are Hausdorff $k$-spaces with $Y$ regular, 
  a set $K \subseteq Y^X$ is compact 
  if and only if
  \begin{enumerate}
  \item $K$ is closed,
  \item the set $K(x)=\{f(x) \mid f \in K\}$ is compact for every $x \in X$,
  \item the set $\bigcap_{f \in K \cap F} f^{-1}(V)$ is open for every
    closed set $F \subseteq Y^X$ and for every open set $V \subseteq Y$.
  \end{enumerate}
\end{quote}
%
Gale didn't assume $Y$ to be a $k$-space and formulated this for the
compact-open topology, but his theorem holds for the exponential
topology if we require $Y$ to be a $k$-space.  Regarding compactness,
we have already mentioned that a Hausdorff space has the same compact
sets as its $k$-coreflection, and that the exponential topology is the
$k$-reflection of the compact-open topology. Although there are more
closed sets in the exponential topology, Gale's argument works with
closedness of $K$ in the exponential topology. This follows
from the general considerations of Kelley~\cite[Chapter 7]{kelley}.

The last condition is a version of equi-continuity.  Because $X$ is
not assumed to be compact, the set $K$ cannot be globally bounded in
any sense, but it is pointwise bounded in the sense of the second
condition.  This gives a characterization of compact subsets of
Kleene--Kreisel spaces of the form $\N^C$ and in particular of
Kleene--Kreisel spaces of pure type, because $\N$ is regular. However,
as discussed in Section~\ref{hyland}, Matthias Schr\"oder has recently
shown that $\N^{\N^\N}$ is not regular, and this justifies the
restriction of our characterizations of exhaustible entire sets to
particular kinds of types in Section~\ref{arzela}.

Notice that when $X=Y=\N$, this amounts to the well known
characterization of compact subsets~$K$ of the Baire space $\N^\omega$
as finitely branching trees. The equi-continuity condition, as in the
case of the Heine--Borel theorem, is superfluous, because any set is
equi-continuous in this case as the topology of the exponent is
discrete. Condition (1) says that the elements of $K$ are the paths of
a tree, and (2) says that the tree is finitely branching, because the
compact subsets of the base space are finite.

\medskip

Lemma~\ref{crucial:generalized} and the remarks preceding it allow one
to consider continuity of functions involving points of a
$k$-space~$X$, open sets and closed sets (using the function space
$\Sierp^X$ and representing open sets and closed sets by their
characteristic functions), and compact sets (using the function space
$\Sierp^{\Sierp^X}$ and representing compact sets by their universal
quantification functionals). We now reformulate Gale's theorem by
expressing condition (3) as a continuous version of a slight
strengthening of condition~(2). 
\begin{thm} \label{gale:modified}
  If $X$ and $Y$ are Hausdorff $k$-spaces with $Y$ regular, a set $K \subseteq
  Y^X$ is compact if and only if
  \begin{enumerate}
  \item $K$ is closed, and
  \item $(K \cap F)(x)$ is compact, continuously in $F$ and $x$,
    for any closed set $F \subseteq Y^X$ and any $x \in X$.
  \end{enumerate}
\end{thm}
\noindent 
The dependence of $(K \cap F)(x)$ in the parameters $F$ and $x$ is
given by the functional
\[
\Phi \colon \Sierp^{Y^X} \times X \to \Sierp^{\Sierp^Y}
\]
defined by
$\Phi(\bar{\chi}_F,x) = \forall_{(K\cap F)(x)},$
where we write
\[
\bar{\chi}_F = \chi_{F^c}.
\]
\begin{proof}
  ($\Rightarrow$): The set $K$ is closed because $Y^X$ is Hausdorff.
  The set $K \cap F$ is compact because $F$ is closed.  Because the
  evaluation map is continuous and because $(K \cap F)(x)$ is the
  continuous image of $K \cap F$ under evaluation at $x$, it is
  compact. To see that $\Phi$ is continuous, let $v \in \Sierp^Y$.  Then
  $v(y)=\top$ for all $y \in (K\cap F)(x)$ $\iff$ $v(f(x))=\top$ for
  all $f \in K\cap F$ $\iff$ $f \in F^c$ or $v(f(x))$ for all $f \in K$.
  Hence \[ \Phi(w,x)=\lambda v.\forall f \in K.w(f) \vee v(f(x)), \]
  where $(\vee) \colon \Sierp \times \Sierp \to \Sierp$ is defined by
  $a \vee b = \top$ iff $a=\top$ or $b=\top$.  Because the
  functional $\forall_K$ is continuous as $K$ is compact, and because
  the category of $k$-spaces is cartesian closed and the above is a
  $\lambda$-definition from continuous maps, $\Phi$ is
  continuous.

\medskip

 ($\Leftarrow$): It suffices to show that Gale's conditions (1)-(3)
 hold.  Condition (1) is the same as ours, and Gale~(2) follows from
 our condition (2) with $F=Y^X$. To prove Gale~(3), let $F \subseteq
 Y^X$ be closed and $V \subseteq Y$ be open. Then the set 
 \[ U = \{ x \in X \mid \Phi(\bar{\chi}_F,x)(\chi_V)=\top \} \]
 is open because $\Phi$ is continuous, and
 \begin{eqnarray*}
   x \in U 
  & \iff & \forall_{(K\cap F)(x)}(v)= \top 
    \iff  \text{$\chi_V(y) = \top$ for all $y \in (K \cap F)(x)$} \\
   & \iff & \text{$\chi_V(f(x)) = \top$ for all $f \in K \cap F$} 
    \iff  \text{$f(x) \in V$  for all $f \in K \cap F$}  \\
   & \iff & {\textstyle x \in \bigcap_{f \in K \cap F} f^{-1}(V)},
 \end{eqnarray*}
 which shows that the set $\bigcap_{f \in K \cap F} f^{-1}(V)$ is the
 same as $U$ and hence is open.
\end{proof}

\begin{defi}
We say \emph{topologically decidable} etc.\ taking
  the continuous versions of Definitions~\ref{decidableonk}
  and~\ref{semidecidable}. \qed
\end{defi}

We now formulate and prove an analogue of this theorem, which replaces
(i)~the Sierpinski space~$\Sierp$ by the boolean domain~$\pBool$,
(ii)~Hausdorff $k$-spaces by Scott domains, (iii)~compact subsets by
topologically exhaustible entire subsets, (iv)~closed subsets by
topologically decidable sets (cf.\ Definitions~\ref{exhaustible:def}
and~\ref{topologically:exhaustible}).  We again apply Gale's theorem,
exploiting Hyland's characterization of the Kleene--Kreisel spaces as
$k$-spaces.  The proof follows the same pattern as that of
Theorem~\ref{gale:modified}, but there are a number of additional
steps. Firstly, using Gale's theorem, we get continuous maps defined
on Kleene--Kreisel spaces.  These are extended to continuous maps on
domains using the Kleene--Kreisel density theorem and Scott's
injectivity theorem, as in Lemma~\ref{clopen:extension}.  (In
Theorem~\ref{exhaustible:arzela}, such an extension will be instead
defined by an algorithm, but still relying on the density theorem.)
Secondly, the set $F$ in condition (2) is closed in
Theorem~\ref{gale:modified} but is neither open nor closed in
Theorem~\ref{exhaustible:arzela:topological}, although it has clopen
shadow, because the Sierpinski space has been replaced by the boolean
domain. To overcome this difficulty, we rely on the following version
of Gale's theorem:

\begin{rem} \label{F:specialcase} An inspection of the proof of Gale's
  theorem shows that it also holds if, in condition (3), the set $F$
  ranges over subbasic closed sets in the compact-open topology:
\begin{enumerate}
\item[$3'$.] the set $\bigcap_{f \in K \cap N(Q,B)} f^{-1}(V)$ is open
  for every compact set $Q \subseteq X$, every closed set $B \subseteq
  Y$, and every open set $V \subseteq Y$.
\end{enumerate}
In one direction this is clear: if condition (3) holds for all closed
$F$, then it holds for~$F=N(Q,B)$. For the other direction, notice
that condition (3) is used only in the ``Lemma'' \cite[page 305]{gale}
for $F$ of this form (the sets $W_x$ in the second last line of that
page, and the set $T$ of page 306). \qed
\end{rem}
Let $D=D_{\sigma}$ and $C = C_{\sigma}$ for an arbitrary type
$\sigma$, and recall the concepts and notation introduced in
Definitions~\ref{semidecidable} and~\ref{topologically:exhaustible}.
\pagebreak[3]
\begin{thm} \label{exhaustible:arzela:topological} An entire set
  $K \subseteq (D \to \pN)$ is topologically exhaustible if and only
  if the following two conditions hold:
\pagebreak[3]
  \begin{enumerate}
  \item[1.] $K$ is topologically co-semi-decidable. 
  \item[2.] The set $(K \cap F)(x)$ is topologically exhaustible for
    any $F$ that is topologically decidable on $K$, and any $x \in D$
    total, continuously in~$F$ and~$x$.
  \end{enumerate}
\end{thm}
Here the dependence of $(K \cap F)(x)$ in $F$ and $x$ is
to be given by a functional
\[
\Gamma \colon ((D \to \pN) \to \pBool) \times D \to ((\pN \to \pBool)
\to \pBool)
\]
such that
$\Gamma(\psi_F,x) = \forall_{(K\cap F)(x)}.$

\begin{proof}
  $(\Rightarrow)$: (1): By Lemma~\ref{lemma:criterion}, the shadow
  $\sK=\rho(K) \subseteq \N^\C$ of $K$ is compact and hence closed.
  Hence the map $\N^\C \to \pBool$ that sends $\f \in \sK$ to~$\bot$
  and $\f \not\in \sK$ to~$\True$ is continuous. By composition with
  the quotient map $\rho \colon T \to \N^C$, where $T=T_{\sigma \to
    \iota}$, we get a map $T \to \pBool$. Because $T$ is dense in $(D
  \to \pN)$ and $\pBool$ is densely injective, the domain
  $D_{\sigma\to\iota}$ under the Scott topology is injective over
  dense embeddings, which means that this map extends to a continuous
  map $(D \to \pN) \to \pBool$.  By construction, this exhibits $K$ as
  a topologically co-semi-decidable subset of $(D \to \pN)$.

  (2): Define $\Gamma(\psi_F,x) = \lambda p.\forall {f \in K}.
  \psi_F(f) \implies p(f(x))$. 
  The result then follows from the fact that the category of Scott
  domains under the Scott topology is cartesian closed, and hence
  functions that are $\lambda$-definable from continuous maps are
  themselves continuous. 

  $(\Leftarrow)$: We apply Gale's theorem to show that the shadow
  $\sK=\rho(K)$ is compact. Then it is topologically exhaustible by
  Lemma~\ref{lemma:criterion}. 

  Gale~(1): If $K$ is topologically co-semi-decidable, then, by
  definition, we have a continuous function $(D \to \pN) \to \pBool$
  that maps $f \in K$ to $\bot$ and $f \not\in K$ to $\True$.  Hence
  $K$ is closed in $T$ because it is the inverse image of the closed
  set $\{\bot\}$ restricted to~$T$.  Because $K$ is entire, it is
  closed under total equivalence by definition, and hence, because
  $\rho \colon T \to \N^C$ is a quotient map, $\sK$ is closed.

  Gale~(2): The assumption gives that for any $x \in D$ total, $K(x)$
  is exhaustible, considering $F=(D \to \pN)$. Because $K$ is entire
  and $x$ is total, $K(x) \subseteq \N$. Hence by
  Lemma~\ref{lemma:criterion}, $K(x)$ is compact in $\N \subseteq
  \pN$.

  \pagebreak[3] Gale~(3): Let $\sF \subseteq \N^C$ be a subbasic open
  set of the form $N(\sQ,V)$ with $\sQ \subseteq C$ compact and $V
  \subseteq \N$ (necessarily) clopen. Then the set $Q=\rho^{-1}(\sQ)$
  is entire and Kleene--Kreisel compact, and hence, by
  Lemma~\ref{lemma:criterion}, it is topologically exhaustible. Also,
  $V$ is a topologically decidable subset of $\pN$. So the predicate
  $p \colon (D \to \pN) \to \pBool$ defined by $p(f) = \forall x \in
  Q. \chi_V(f(x))$ is continuous and defined on~$K$, and $p=\psi_F$
  for $F = T \cap p^{-1}(\True)$.  Now define $u \colon D
  \to \pBool$ by \[ u(x) = \Gamma(\psi_F,x)(\chi_V). \] Then
  $u$ is continuous and
  \[
  u(x) =  \forall_{(K\cap F)(x)}(\chi_V) 
  = \forall f \in K \cap F.\chi_V(f(x)).
  \]
  Hence the set $U = u^{-1}(\True) = \{ x \in T \mid \forall f \in K
  \cap F.\chi_V(f(x)) =\True \}$ is open. Therefore
  its shadow $\bigcap_{\f \in \sK \cap \sF} \f^{-1}(V)$ is open,
  because it is closed under total equivalence and because $\rho$ is a
  quotient map. 
\end{proof}

\subsection{Computational version}
\label{arzela}


At this stage of our investigation, such a characterization is
available only for certain types, which include pure types, and for
entire sets (for the reasons explained in
Section~\ref{arzela:compact}). Let $D=D_{\sigma}$ and $C = C_{\sigma}$
for an arbitrary type $\sigma$.  We establish the computational
version of Theorem~\ref{exhaustible:arzela:topological}.
\pagebreak[3]
\begin{thm} \label{exhaustible:arzela}
  An entire set $K \subseteq (D \to \pN)$ is exhaustible if and only if 
  the following two conditions hold:
  \begin{enumerate}
  \item[1.] $K$ is co-semi-decidable.
  \item[2.] The set $(K \cap F)(x)$ is exhaustible for any $F$
    decidable on $K$, and any $x \in D$ total, uniformly in~$F$
    and~$x$.
  \end{enumerate}
  Moreover, the equivalence is uniform. 
\end{thm}

\medskip

A few remarks are in order before embarking into the proof.
The claim holds, with the same proof, if
conditions (1) and (2) are replaced by any of the following
conditions, respectively:
\begin{enumerate}
\item[$1'$.] $K$ is topologically co-semi-decidable.
\item[$1''$.] $K$ has closed shadow.
\item[$1'''$.] The shadow of $K$ is closed in the topology of
  pointwise convergence.
\item[$1''''$.] $K$ has compact shadow.
\item[$2'$.]  The set $(K \cap F_n^\alpha)(x)$ is exhaustible, uniformly
    in $n \in \pN$, $\alpha \in \pN^\myomega$ and $x \in D$ total.
\end{enumerate}
Recall (proof of Theorem~\ref{ex:main}) that
we defined
\[
F_n^\alpha = \{ f \in (D \to \pN) \mid \forall i<n. f(d_i) = \alpha_i
\}. 
\]

In the formulation of the theorem, the fact that
conditions~(1) and~(2) uniformly imply the exhaustibility of $K$ is in
principle given by a computable functional of type

{\footnotesize
\[
\underbrace{\overbrace{((D \to \pN) \to \Sierp)}^{\chi_{K^c}}}_{\text{condition 1}} 
\times \underbrace{(\overbrace{((D \to \pN) \to \pBool)}^{\psi_F} \times \overbrace{D}^{x} \to  \overbrace{((\pN \to \pBool) \to \pBool)}^{\forall_{K \cap F(x)}})}_{\text{condition 2}}
\to \underbrace{\overbrace{(((D \to \pN) \to \pBool) \to \pBool)}^{\forall_K}}_{\text{conclusion}}.
\]}

However, the computational information given by condition (1) is not
used in the construction of the conclusion (although the topological
information is used in its correctness proof). Moreover, the
information given by condition (2) is not fully used in the
construction. Replacing it by ($2'$) we get
\[
\underbrace{(\overbrace{\pN^\myomega}^{\alpha} \times \overbrace{\pN}^{n} \times \overbrace{D}^{x}) \to \overbrace{((\pN \to \pBool) \to \pBool)}^{\forall_{K\cap F_n^\alpha(x)}})}_{\text{condition $2'$}}
\to \underbrace{(((D \to \pN) \to \pBool) \to \pBool)}_{\text{conclusion}}.
\]
Additionally the pair $\alpha,n$ is really coding a finite sequence,
and, as we have seen, exhaustible sets of natural numbers are uniformly
equivalent to finite enumerations of natural numbers. Hence the
above can be written as
\[
\underbrace{((\overbrace{\pN^*}^{\alpha,n} \times \overbrace{D}^{x}) \to \overbrace{\pN^*}^{\forall_{K\cap F_n^\alpha}(x)})}_{\text{condition $2'$}}
\to \underbrace{(((D \to \pN) \to \pBool) \to \pBool)}_{\text{conclusion}}.
\]
Therefore the above characterization reduces the type level of
$\forall_K$ by two.

\medskip

The last step of the proof of this theorem mimics topological proofs
of Arzela--Ascoli type theorems (which we haven't included): to show
that $K \subseteq Y^X$ is compact under assumptions such as those of
Gale's theorem (Section~\ref{arzela:compact}), one first concludes
that $\prod_{x \in X} K(x)$ is compact by the Tychonoff theorem, then
shows that the relative topology of $K$ is the topology of pointwise
convergence, and that it is pointwise closed, and hence concludes that
it is homeomorphically embedded into the product as a closed subset,
and therefore that it must be compact. In the proof below, we have
replaced the Tychonoff theorem by its countable computational version
given by Theorem~\ref{searchable:tychonoff}, using a dense sequence of
the exponent. The first steps of the proof are needed in order to
make this replacement possible, and they are modifications of the
constructions developed in Section~\ref{characterization}.

\begin{proof}
($\Rightarrow$) (1): Theorem~\ref{ex:main}.

\medskip

(2): Define $\forall_{K \cap F}(p) = \forall {f \in K}.  \psi_F(f)
\implies p(f(x))$.

\medskip

($\Leftarrow$): By Theorem~\ref{exhaustible:arzela:topological}, the
set $K$ is topologically exhaustible, and hence is Kleene--Kreisel
compact by Lemma~\ref{lemma:criterion}. This compactness conclusion is
our only use of Theorem~\ref{exhaustible:arzela:topological} in this
proof. We apply this to establish the correctness of the algorithms
defined below.

Define $P \colon (D \to \pN) \to \pN^\myomega$ by $P(f)(i) = f(d_i)$, as
in the proof of Theorem~\ref{ex:main}, where $d \in D^\myomega$ is a
computable dense sequence, and let $L=P(K)$.  Because
$F_n^\alpha$ is decidable on $K$, the set $K_n^\alpha = K \cap F_n^\alpha$ is
exhaustible by Proposition~\ref{prop:intersec}, and $K_n^\alpha(x)$ is
exhaustible uniformly in $\alpha$, $n$ and $x \in K$ by
Proposition~\ref{image} applied to evaluation at~$x$.
Now modify the definition of $r \colon \pN^\myomega \to \pN^\myomega$
given in Theorem~\ref{ex:main} as follows:
\[
 r(\alpha)(n) = 
 \begin{cases}
   \alpha_n & \text{if $\exists y \in K_n^{r(\alpha)}(d_n). \alpha_n = y$,} \\
   \mu y. \exists y' \in K_n^{r(\alpha)}(d_n). y = y' & \text{otherwise.}
 \end{cases}
\]
Then $r$ is computable, and satisfies
\[
 r(\alpha)(n) = 
 \begin{cases}
   \alpha_n \qquad \qquad \quad \text{if $\exists f \in K. f(d_n) = \alpha_n \wedge \forall i < n. f(d_i) = r(\alpha)(i) $,} \\
   \mu y. \exists f \in K. f(d_n) = y \wedge \forall i < n. f(d_i) = r(\alpha)(i)  \qquad\,\,\,
   \text{otherwise.}
 \end{cases}
\]
Hence it also satisfies
\[
 r(\alpha)(n) = 
 \begin{cases}
   \alpha_n & \text{if $\exists \beta \in L. \beta_n = \alpha_n \wedge \beta_i =_n r(\alpha) $,} \\
   \mu y. \exists \beta \in L. \beta_n = y \wedge \beta =_n r(\alpha)  & 
   \text{otherwise.}
 \end{cases}
\]
This shows that $r=r_L$ for $r_L$ as defined in Theorem~\ref{ex:main}.
But notice that, although the second and third equations hold, the
algorithm is not the same as in Theorem~\ref{ex:main}.  In fact, the
second and third equations don't establish computability of $r$,
because exhaustibility of $K$ and $L$ are not known at this stage of
the proof. In any case, the last equation shows that $r$ exhibits $L$
as a retract up to total equivalence, using the fact that $L$, being
the continuous $P$-image of $K$, is topologically exhaustible and
hence is Kleene--Kreisel compact, as in Theorem~\ref{ex:main}

Similarly, modify the definition of $E \colon \pN^\myomega \to (D \to
\pN)$ in Theorem~\ref{ex:main} as follows:
\[
E(\alpha)(x) = 
  \mu y. \exists y' \in K_n^{r(\alpha)}(x).y = y',
\]
where
  $n$ is the least number such that $\forall y,y' \in
  K_n^{r(\alpha)}(x).y=y'.$
Because this condition is equivalent to
  $\forall f,f' \in K_n^{r(\alpha)}.f(x)=f'(x),$
  such a number exists by Lemma~\ref{c} and the compactness of the
  shadow of~$K$. By uniform exhaustibility of the
  set~$K_n^{r(\alpha)}(x)$, this can be found uniformly in $\alpha$
  and $x$, and hence $E$ is computable. Moreover, although the
  definition of $E$ is not the same, as before, we again have $E=E_L$
  for $E_L$ defined as in Theorem~\ref{ex:main}.

  Finally, because $K = K \cap F$ for $F=(\D \to \pN)$, the set $K(x)$
  is exhaustible uniformly in $x \in D$ total, and hence the set $M =
  \prod_i K(d_i) \subseteq \pN^\myomega$ is searchable uniformly in $x
  \mapsto \forall_{K(x)}$.  In fact, each $K(d_i)$ is searchable
  uniformly in $i$, by Theorem~\ref{ex:main}, and hence $M$ is
  searchable by Theorem~\ref{searchable:tychonoff}. Now $L \subseteq
  M$ and hence the entire $r$-image of $M$ is $L$, and hence $L$ is
  searchable by Proposition~\ref{image:bis}. In turn $K$ is the entire
  $E$-image of $L$ and hence is also searchable.  Therefore it is
  exhaustible.
\end{proof}

Notice that the proof actually concludes that $K$ is searchable, and
hence we could have formulated the theorem as: An entire set $K
\subseteq (D \to \pN)$ is searchable iff $K$ is co-semi-decidable and
the set $(K \cap F)(x)$ is exhaustible for any $F$ decidable on $K$,
and any $x \in D$ total, uniformly in~$F$ and~$x$.  But this
strengthening of the theorem follows from the given formulation and
the results of Section~\ref{characterization}. However, we could have
included the above theorem, with the stronger formulation, before
Section~\ref{characterization} and then derived the results of that
section as a corollaries. But we feel that the developments of both
sections become more mathematically transparent with the current
organization of the technical material.

\section{Technical remarks, further work, applications and directions} \label{technical}

We now discuss some technical aspects of the above development,
announce some results that we intend to report elsewhere, and discuss
potential applications and directions for future work in this field.

\subsection{Analysis of the selection functional given by the product functional}

\newcommand{\eval}{\operatorname{eval}}

Using course of values induction, one easily sees that a functional
\licsmath{\Pi \colon ((\D \to
  \pBool) \to \D)^\omega \to ((\D^\omega \to \pBool) \to \D^\omega)}
satisfies the equation of Definition~\ref{product:functional} if and
only if it satisfies the equation
\[
\text{$\Pi(\e)(p) = x_0 * \Pi(\e')(p_{x_0})$ where $x_0 = \e_0(\lambda x. p_x(\Pi(\e')(p_x)))$}
\]
and where $p_x(\alpha) = p(x * \alpha)$ and $\e'_i = \e_{i+1}$.  Now
define a selection function \[ \e_2 \colon (\pBool \to \pBool) \to
\pBool \] for $2 \subseteq \pBool$ by
\[
\e_2(p) = p(\True) = \If p(\True) \Then 1 \Else 0
\]
and a selection function \[ \d \colon (\pBool^\omega \to \pBool) \to
\pBool^\omega\] for the Cantor space $2^\omega \subseteq \pBool^\omega$ by
\[
\d = \Pi(\lambda i.\e_{2}).
\]
Then $\d$ satisfies the equation
\[
\text{$\d(p) = x_0 * \d(p_{x_0})$ where $x_0 = p_1(\d(p_1))$.}
\]
An interesting aspect of this selection function for the Cantor space
is that it doesn't perform case analysis on the value of~$p$, and so,
in some sense, it doesn't work by trial and error.

In order to understand this, first notice that the above recursive
definition of $\d$ makes sense if the domain of booleans is replaced
by any domain $T$ with an element $1 \in T$:
\[
\d \colon (T^\omega \to T) \to \T^\omega.
\]
We consider the case in which~$T=T_{\omega}$ is the domain of possibly
non-well-founded $\omega$-branching trees with leaves labelled by $1$.
We define this as the canonical solution of the domain equation
\[ T \cong \{\True\} + T^\omega,\] where the sum is lifted.  Thus, a
tree is either $\bot$, or else a leaf $1$, or else an unlabelled root
followed by a forest of countably many trees.  Denote the
canonical isomorphism by
\[
\{\True\} + T^\omega  \stackrel{[1,P]}{\longrightarrow}  T.
\]
Then $P \colon T^\omega \to T$ and the forest $\d(P) \in T^\omega$
gives a general formula for solving $p(\alpha)=1$ with $\alpha$
ranging over~$2^\omega$.  In fact, for any given $p \in (\pBool^\omega
\to \pBool)$, define an evaluation function $\eval = \eval_p \colon T
\to \pBool$ by
\begin{eqnarray*}
   \eval(\True) & = & \True, \\
   \eval(P(\alpha)) & = & p(\lambda i.\eval(\alpha_i)).
\end{eqnarray*}
Equivalently, $\eval$ is the unique homomorphism
from the initial algebra $[1,P] \colon \{1\} + T^\omega\to T$ to the
algebra $[1, p] \colon \{1\} + \pBool^\omega\to \pBool$.  Hence the
solution $\alpha$ of the equation $p(\alpha)=1$ is given by evaluating
the general solution $\d(P)$ at $p$:
\[ \alpha_i = \eval_p(\d(P)(i)).\] We illustrate this with finite
forests.  Any $p \in (\pBool^\omega\to\pBool)$ defined on $2^\omega$
is uniformly continuous, and hence of the form
$p(\alpha)=q(\alpha_0,\dots,\alpha_{n-1})$ for some~$n$ and for $q
\colon \pBool^n \to \pBool$ defined by this equation.  Now consider
the domain $T=T_n$ of $n$-branching trees,
\[ T \cong \{\True\} + T^n,\] and denote the canonical isomorphism by
\[
\{\True\} + T^n  \stackrel{[1,Q]}{\longrightarrow}  T.
\]
To make sense of the above definition of $\d$ for this choice of $T$,
define $x * (\alpha_0, \dots, \alpha_{n-2},\alpha_{n-1})= (x,
\alpha_0, \dots, \alpha_{n-2}) \in T^n$ for $x \in T$ and $\alpha \in
T^n$. We tabulate some forests, which grow doubly exponentially, but
only exponentially if auxiliary variables are used to denote common
subtrees (corresponding to the variable $x_0$ in the recursive
definition of~$\d$):
\[
\begin{array}{ll}
n & \d(Q) \\ \hline
1 &  Q(1) \\
2 & (Q(1,Q(1,1)), Q(Q(1,Q(1,1)),1))\\
3 & (x_0, x_1, x_2)
\end{array}
\]
where
\begin{eqnarray*}
x_0 & = & \text{$Q(\True,y,Q(\True,y,\True))$ with $y  = Q(\True,\True,Q(\True,\True,\True))$,}\\
x_1 & = & Q(x_0,\True,Q(x_0,\True,\True)), \\
x_2 & = & Q(x_0,x_1,\True).
\end{eqnarray*}
In order to find $(x_0, x_1, x_2) \in \pBool^3$ such that $q(x_0, x_1,
x_2)=\True$ holds, we substitute $q$ for $Q$ in the above equations,
compute $(x_0, x_1, x_2)$, and check whether $q(x_0, x_1, x_2)=\True$
holds. If it does, then we have found a solution (in fact the largest
in the lexicographic order), and otherwise we conclude that there is
no solution.  Thus, the forest $\d(Q)$ gives a closed formula for
solving the equation $q(\alpha)=1$, and telling whether there is a
solution, composed only from $q$ and the constant~$1$.  To solve
$q(\alpha)=0$, just replace $1$ by $0$ in the formula.

\subsection{Solution of equations with exhaustible domain}

By definition, a set $K \subseteq D$ is searchable iff for every
predicate $p \in (D \to \pBool)$ defined on $K$ one can find $x_0 \in
K$, uniformly in $p$, such that if the equation $p(x)=\True$ has a
solution $x \in K$, then $x=x_0$ is a solution.  We first observe that
this is equivalent to requiring that for every function $f \in (D \to
\pN)$ defined on $K$ and any total $y \in \pN$ one can find $x_0 \in
K$, uniformly in $f$ and $y$, such that if the equation $f(x)=y$ has a
solution $x \in K$, then $x=x_0$ is a solution. For one direction,
consider the predicate $p(x) = (f(x) == y)$, and, for the other,
consider the natural inclusion of $\pBool$ into $\pN$ (which is the
identity under our notation). Clearly, this generalizes from $\pN$ to
any domain $E=D_{\sigma}$ with $\sigma$ discrete in the sense of
Definition~\ref{discrete:compact}. But notice that in this case the
equation has to be written in the form $f(x) \sim y$.

\pagebreak[4]
It is natural to ask whether this generalizes to functions $f \in (D
\to E)$ with $E$ arbitrary. But it is known that, in general, if an
equation has more than one solution, it is typically not possible to
algorithmically find some solution~\cite{beeson}.  We announce the
following result: \pagebreak[3]
\begin{quote} \em Let $D=D_{\sigma \to \iota}$ and $E=D_{\tau \to
    \iota}$ for types $\sigma$ and $\tau$, let $K \subseteq D$ be an
  exhaustible entire set, $F \in (D \to E)$ be total and $g_0 \in E$
  be total.
\begin{enumerate}
\item If the equation $F(f) \sim g_0$ has a solution $f \in K$, unique
  up to total equivalence, then some $f_0 \sim f$ is computable,
  uniformly in $F$, $g_0$ and any universal quantification functional
  for~$K$.
\item It is semi-decidable whether $F(f) \sim g_0$ doesn't have a
  solution $f \in K$, with the same uniformity condition.
\end{enumerate}
\end{quote}
In order to establish this, we prove the following generalization of
Lemma~\ref{c}: Let $K_n \subseteq D_{\sigma \to \iota}$ be a sequence
of entire sets such that $K_n \supseteq K_{n+1}$ and that $\bigcap_n
K_n$ is the equivalence class of some total~$f$. Then one can find a
computable total function $f_0 \sim f$, uniformly in any sequence of
universal quantification functionals for~$K_n$.

This is not very useful in computable analysis via representations,
because typically uniqueness, when it holds, is only up to
equivalence of representations rather than total equivalence. But we
do have a corresponding result for equations involving real numbers.
In light of the following, it is natural to ask whether there is a
further corresponding result for real valued functions of real
variables.

\subsection{An exhaustible set of analytic functions}
An application of the exhaustibility of the Cantor space to the
computation of definite integrals and function maxima has been given
by Simpson~\cite{simpson:integration}. A generalization of this is
developed by Scriven~\cite{scriven}. We consider computation with real
numbers via admissible Baire-space
representations~\cite{weihrauch:analysis} and domain
representations~\cite{blanck}.  For any $x \in \I = [-1/2,1/2]$ and
any sequence $a \in [-b,b]^\omega$, the Taylor series $\sum_n a_n x^n$
converges to a number in the interval $[-2b,2b]$.
We announce the following example of a searchable, and hence
exhaustible, set:
\begin{quote}
  \em For any real number $b > 0$, the set $A = A_b$ of analytic
  functions $f \colon \I \to \R$
\[
\text{$f(x)=\sum_n a_n x^n$ \qquad with \quad $a \in [-b,b]^\omega$}
\]  
has a searchable set of representatives, uniformly in $b$.
\end{quote}
In our proof of this, we argue that any $f \in A$ can be computed
uniformly in its Taylor coefficients and use the fact that
 $[-b,b]^\omega$ has a searchable set of representatives.
This can be used to deduce that:
\pagebreak[3]
  \begin{enumerate}
  \item The Taylor coefficients of any $f \in A$ can be computed
        uniformly in~$f$.

  \item The distance function $d_A \colon \R^\I \to \R$ defined by
    \[ d_A(g) = \min \{ d(f,g) \mid f \in A \} \] is computable (cf.\
    Bishop's notion of locatedness~\cite{MR36:4930,bishop:bridges}).

  \item For any $f \in \R^\I$, it is semi-decidable, uniformly in $f$,
        whether $f \not \in A$. 
  \end{enumerate}

\subsection{Peano's theorem}

This celebrated theorem asserts that certain differential equations
have solutions, but without indicating what the solutions might look
like.  Its proofs are typically based on the Arzela--Ascoli theorem,
and proceed by applying Euler's algorithm to produce a sequence of
approximate solutions. In general, however, this sequence is not
convergent, but, by an application of compactness, there is a
convergent subsequence, although no specific example is exhibited by
this argument, which is then easily seen to produce a solution of the
equation.  It is therefore natural to ask whether our tools could be
applied to compute unique solutions of such differential equations
under suitable assumptions.  Here the goal is not to obtain a usable
algorithm, but rather to understand the classical proof from a
computational perspective in connection with the notion of
exhaustibility and its interaction with the notion of compactness and
with the Arzela--Ascoli theorem.

\subsection{Uncountable products of searchable sets}

It is natural to ask whether the countable product
theorem~\ref{searchable:tychonoff} can be generalized to uncountable
index sets.  This question is pertinent in view of well known
constructive versions of the Tychonoff theorem in locale
theory~\cite{MR641111} and formal topology~\cite{MR1150923}, which
don't restrict the cardinality of the index set.  However, this seems
unlikely in the realm of Kleene--Kreisel higher type computability
theory.  Consider the case in which the index set is the Cantor space.
By the classical Tychonoff theorem, the product of $\Bool^\omega$-many
copies of~$\Bool$ is compact.  This product could be written
as~$\Bool^{\Bool^\omega}$. But this notation in higher-ype
computation is interpreted as a function space, and in the category of
Kleene--Kreisel spaces one has $\Bool^{\Bool^\omega}\cong \N$, because
the base is discrete and the exponent is compact (cf.\
Theorem~\ref{thm:discrete:compact}).  This phenomenon in fact also
takes place in the categories of locales~\cite{hyland:functionspaces}
and topological spaces~\cite{escardo:barbados}. In the Tychonoff
theorem for locales or spaces, the indices form a set or equivalently a
discrete space. But a discrete Kleene--Kreisel space is countable (and
more generally a discrete QCB space is countable).

\subsection{Totality of the product functional and bar recursion}

\newcommand{\CBR}{\operatorname{CBR}}
\newcommand{\MBR}{\operatorname{MBR}} In
Theorem~\ref{searchable:tychonoff} we constructed a computable
functional $\Pi$ such that $\Pi(\e)(p) \in \prod_i K_i$
whenever~$\e_i$ is a selection functional for a set $K_i$ and $p$ is
defined on $\prod_i K_i$. It is natural to ask whether the
functional~$\Pi$ is actually total. Paulo Oliva has shown that this is
indeed the case (personal communication). Moreover, he has observed
that if the type of booleans is replaced by the type of natural
numbers, our recursive definition of $\Pi$ still makes sense and that
it also gives rise to a total functional, which he calls $\CBR$
(course-of-values bar recursion). He additionally proved that $\CBR$
is primitively recursively inter-definable with the modified bar
recursion functional $\MBR$ defined in~\cite{berger:oliva:mbr}. We are
currently investigating together the ramifications of these
observations.

\subsection{Alternative notions of exhaustibility}

If one is interested only in total functionals and sets of total
elements, it is natural formulate the following alternative notion of
exhaustibility: A set $K \subseteq D$ is \emph{entirely exhaustible}
if it is entire and there is a total computable functional $\forall_K
\colon (D \to \pBool) \to \pBool)$ such that for every total $p \in (D
\to \pBool)$ one has $\forall_K(p)=\True$ iff $p(x)=\True$ for all $x
\in K$. Because, as we have seen, non-empty, exhaustible entire sets
are computable retracts, it follows that any exhaustible entire set is
entirely exhaustible. The converse fails (but see the next paragraph),
because e.g.\ any dense subset of the Cantor space is entirely
exhaustible using Berger's algorithm and the fact that any total
predicate is uniquely determined, up to total equivalence, by its
behaviour on a dense of set of total elements.

Moreover, when one is only interested in total functions and total
elements, it is perhaps more natural to work with Kleene--Kreisel
spaces directly, without the detour via domains, e.g.\ defined as
$k$-spaces.  QCB spaces are a natural and general setting for such
considerations~\cite{MR2328287,MR1948051}.  One might say that a
subset $K$ of a space~$X$ is \emph{totally exhaustible} if there is a
computable functional $\forall_K \colon (X \to \Bool) \to \Bool$ such
that for every $p \in (X \to \Bool)$, we have that
$\forall_K(p)=\True$ iff $p(x)=\True$ for all $x \in K$.  When e.g.\
$X=\N^\N$, total exhaustibility of $K \subseteq X$ doesn't entail
compactness of~$K$, again considering the example of a dense subset of
the Cantor space. But Matthias Schr\"oeder (personal communication in
2006) proved that if $X$ is a QCB space which is the sequential
coreflection of a zero-dimensional Hausdorff space, then any totally
exhaustible \emph{closed} set $K\subseteq X$ is compact. This includes
the case in which $X$ is a Kleene--Kreisel space. Using this and the
above observations, one can show that, as far as higher-type
computation with total continuous functionals is concerned, the
notions of exhaustibility and total exhaustibility agree for closed
sets.

\subsection{A unified type system for total and partial computation}

As we have already discussed, Kleene--Kreisel spaces and Ershov--Scott
domains live together in the cartesian closed category of compactly
generated spaces, and in fact in the subcategory of QCB spaces.
Additionally, the inclusions of $k$-spaces and of Ershov--Scott
domains into these categories preserve the cartesian-closed
structure~\cite{escardo:lawson:simpson,MR2328287}.  Hence total and
partial higher-type functionals coexist in the same cartesian closed
category.  One can envisage a higher-type system that simultaneously
incorporates, but explicitly distinguishes, total and partial objects,
and corresponding PCF-style formal systems. Among the formation rules
one can have two types for the natural numbers, with and without
$\bot$, and it would make sense to stipulate that $\sigma \to \tau$ is
a partial type whenever $\sigma$ is any type and $\tau$ is a partial
type, and that $\sigma \to \tau$ is a total type when both $\sigma$
and $\tau$ are total types. In its simplest form, such a language
could include G\"odel's system $T$ for total types and PCF for partial
types. Such a formalism would have simplified, and made more
transparent, much of the development concerning exhaustible sets of
total elements, where we could have benefited from functionals that
take total inputs and produce potentially partial outputs.  In
particular, all the technical considerations of total equivalence and
shadows could have been avoided in this way, making the development
more transparent. Such functionals are actually total, but their
construction uses modes of definition that belong to the realm of
partial computation. The system-$T$ fragment could be further extended
with total computable functionals such as bar recursion and some of
those developed here, once one has shown they are indeed total.

\subsection{Time complexity of exhaustive search}

In the paper~\cite{escardo:lics07}, we report some surprisingly fast
experimental results, which serve to counteract an impression that
might be gained from the technical development that the algorithms
presented would be essentially intractable and of purely theoretical
interest. Moreover, that paper formulates run-time conjectures that
provide examples of questions that one would like to be able to treat
rigorously and that are potentially useful as target problems for work
in higher-type complexity theory.  The conjectures express the run
time in terms of the modulus of uniform continuity of the input
predicate on the exhaustible set, and hence topology seems to play a
role in higher-type complexity too.

It might be possible to apply our search algorithms to practical
problems, e.g.\ in real analysis and in program verification.  But it
is more likely that, in order to obtain feasible algorithms, such
applications will need to rely on the development of particular
algorithms for particular kinds of infinite search tasks, perhaps
inspired or guided by the general algorithms we have developed, but in
any case needing new insights and techniques.  In fact, this is
already the case for finite search problems, as is well known.  But
the fast examples reported in~\cite{escardo:lics07} do highlight that
the task of obtaining particular search algorithms that are efficient
for particular kinds of infinite search problems of interest is a
direction of research that deserves attention and is likely to be
fruitful, and that a study of feasible infinite search problems cries
to be carried out.

\subsection{A fast product functional}

We have just discussed that one should look for efficient search
algorithms for specialized problems.  But it is still interesting to
ask how fast a general infinite search algorithm can be.  We don't
know the answer, but we report an algorithm that outperforms all the
algorithms applied for the experimental results
of~\cite{escardo:lics07}, and whose theoretical run-time behaviour
remains to be investigated.

\newcommand{\Root}{\operatorname{root}}
\newcommand{\Left}{\operatorname{left}}
\newcommand{\Right}{\operatorname{right}}
\newcommand{\branch}{\operatorname{branch}} 

We regard an infinite sequence~$t$ as an infinite binarily
branching tree with the elements of the sequence organized in a
breadth-first manner: the root is $t_0$, and the left and right
branches of the node $t_n$ are $t_{2n+1}$ and $t_{2n+2}$.  With this
in mind, define functions
\begin{quote}
  $\Root \colon E^\myomega \to E$, \,\, $\Left,\Right \colon E^\myomega \to
  E^\myomega$, \,\, $\branch \colon E \times E^\myomega \times E^\myomega \to
  E^\myomega$
\end{quote}
by
\[
\Root(t) = t_0, \qquad
\Left (t) = \lambda i. t_{2i + 1}, \qquad
\Right(t) = \lambda i. t_{2i + 2},
\]
\[
\branch(x,l,r) = \lambda i.
\begin{cases}
x & \text{if $i = 0$,} \\
l_{(i-1)/2} & \text{if $i$ is odd,} \\
r_{(i-2)/2} & \text{otherwise.}
\end{cases}
\]
Then, for $x \in E$ and $l,r \in E^\myomega$,
\[
\Root(\branch(x,l,r))=x, \quad
\Left(\branch(x,l,r))= l, \quad
\Right(\branch(x,l,r))=r, 
\]
\[
\branch(\Root(t),\Left(t),\Right(t))=t.
\]
Our experimentally faster product algorithm is then recursively defined by
\[
\Pi(\e)(p) = \branch(x_0,l_0,r_0)
\]
where
\[
  \e_{\Root}  =  \Root(\e), \qquad
  \e_{\Left}  =  \Pi(\Left(\e)), \qquad
  \e_{\Right}  =  \Pi(\Right(\e)),
\]
\[  \exists_{\Left}(p)  = p(\e_{\Left}(p)), \qquad
  \exists_{\Right}(p)  =  p(\e_{\Right}(p)),
\]
\begin{eqnarray*}
x_0 & = & \e_{\Root}(\lambda x.\exists_{\Left} \,l.\,\exists_{\Right} \,r.\,p(\branch(x,l,r))),
\\  l_0 & = & \e_{\Left}(\lambda l.\exists_{\Right} \,r.\,p(\branch(x_0,l,r))),
\\  r_0 & = & \e_{\Right}(\lambda r.p(\branch(x_0,l_0,r)).
\end{eqnarray*}
The idea is that treating sequences as trees reduces some linear
factors to logarithmic factors (very much like in the well-known
heap-sort algorithm).

\subsection{Operational perspective}
An advantage of the proof of Theorem~\ref{searchable:tychonoff}
sketched in~\cite{escardo:lics07} is that it can be directly
interpreted in the operational
setting~\cite{escardo:barbados,escardo:ho}.  The proofs of the other
results of Section~\ref{building} are also easily seen to work in the
above operational setting.  But a development of operational
counter-parts for those of later sections is left as an open problem.
This requires an operational reworking of the topological
Section~\ref{criteria}, which seems challenging.

\section{Concluding remark on the role of topology} 
\label{conclusion}

The algorithms developed in this work have purely computational
specifications, which allow them to be applied without knowledge of
specialized mathematical techniques in the theory of computation.
However, the correctness proofs of some of the algorithms crucially
rely on topological techniques. In this sense, this work is a genuine
application of topology to computation: theorems formulated in the
language of computation, proofs developed in the language of topology.

But there is another sense in which topology proves to play a crucial
role.  Compact sets in topology are advertised as sets that behave, in
many important respects, as if they were finite.  Then exhaustively
searchable sets \emph{ought} to be compact.  And compact sets are
known to be closed under continuous images and under finite and
infinite products.  Moreover, for countably based Hausdorff spaces,
they are the continuous images of the Cantor space.  Hence searchable
sets \emph{ought} to have corresponding closure properties and
characterization, which is what this work establishes, among other
things, \emph{motivated} by these considerations.  Thus, in a more
abstract level, topology is applied as a paradigm for discovering
unforeseen notions, algorithms and theorems in computability theory.

\nocite{smyth:topology}
\bibliographystyle{plain}
\bibliography{exhaustive-journal}

\vskip-50 pt




\end{document}